\documentclass [a4paper]{article}
\usepackage{verbatim}                       

\usepackage{tikz-cd}

\usepackage{listings}
\usepackage[plain]{fancyref}

\usepackage{amsmath, amsthm}         
\usepackage{amssymb}
\usepackage{enumerate}                      

\usepackage{accents}
\usepackage[section]{placeins}
\usepackage{enumitem}
\usepackage[hscale=0.65, vscale=0.73]{geometry}
\usepackage{thm-restate}

\newcommand{\mf}[1]{\mathfrak{#1}}
\newcommand{\N}{\mathbb{N}}
\newcommand{\Z}{\mathbb{Z}}
\newcommand{\R}{\mathbb{R}}
\newcommand{\T}{\mathbb{T}}
\newcommand{\C}{\mathbb{C}}
\DeclareMathOperator*{\supp}{supp}
\DeclareMathOperator{\Diff} {Diff}
\newcommand{\SET}[1]{ \left\{ \,#1\,\right\}}

\newcommand{\olpi}{\overline{\pi\vphantom{*}}}

\newcommand{\dquotes}[1]{``#1''}
\newcommand{\circled}[1]{\accentset{\circ}{#1}}
\newcommand{\mrm}[1]{\mathrm{#1}}
\newcommand{\der}{\text{der}}
\newcommand{\mc}{\mathcal}
\newcommand{\g}{\mathfrak{g}}
\newcommand{\gl}{\mathfrak{gl}}
\newcommand{\h}{\mathfrak{h}}
\newcommand{\n}{\mathfrak{n}}
\newcommand{\mcD}{\mathcal{D}}
\newcommand{\X}{\mathcal{X}}
\newcommand{\U}{\mathrm{U}}
\newcommand{\PU}{\mathrm{PU}}
\newcommand{\pu}{\mathfrak{pu}}
\newcommand{\St}{\mathrm{St}}
\newcommand{\Ad}{\mathrm{Ad}}
\newcommand{\restr}[2]{\left. #1\right|_{#2}}
\DeclareMathOperator*{\st}{\; : \;}
\newcommand{\Hom}{\mathrm{Hom}}

\newcommand{\Aut}{\mathrm{Aut}}
\newcommand{\ct}{\mathrm{ct}}
\newcommand{\id}{\mathrm{id}}

\theoremstyle{plain}
\newtheorem{theorem}{Theorem}[section]
\newtheorem*{theoremA}{Theorem A}
\newtheorem*{theoremB}{Theorem B}

\theoremstyle{definition}
\newtheorem{definition}[theorem]{Definition}

\theoremstyle{plain}
\newtheorem{lemma}[theorem]{Lemma}
\newtheorem{proposition}[theorem]{Proposition}
\newtheorem{corollary}[theorem]{Corollary}

\theoremstyle{remark}
\newtheorem{remark}[theorem]{Remark}

\theoremstyle{definition}

\numberwithin{equation}{section}

\newcommand*{\fancyrefthmlabelprefix}{thm}
\fancyrefaddcaptions{english}{%
  \providecommand*{\frefthmname}{theorem}%
  \providecommand*{\Frefthmname}{Theorem}%
}
\frefformat{plain}{\fancyrefthmlabelprefix}{\frefthmname\fancyrefdefaultspacing#1}
\Frefformat{plain}{\fancyrefthmlabelprefix}{\Frefthmname\fancyrefdefaultspacing#1}

\newcommand*{\fancyreflemlabelprefix}{lem}
\fancyrefaddcaptions{english}{%
  \providecommand*{\freflemname}{lemma}%
  \providecommand*{\Freflemname}{Lemma}%
}
\frefformat{plain}{\fancyreflemlabelprefix}{\freflemname\fancyrefdefaultspacing#1}
\Frefformat{plain}{\fancyreflemlabelprefix}{\Freflemname\fancyrefdefaultspacing#1}

\newcommand*{\fancyrefproplabelprefix}{prop}
\fancyrefaddcaptions{english}{%
  \providecommand*{\frefpropname}{proposition}%
  \providecommand*{\Frefpropname}{Proposition}%
}
\frefformat{plain}{\fancyrefproplabelprefix}{\frefpropname\fancyrefdefaultspacing#1}
\Frefformat{plain}{\fancyrefproplabelprefix}{\Frefpropname\fancyrefdefaultspacing#1}

\newcommand*{\fancyrefcorlabelprefix}{cor}
\fancyrefaddcaptions{english}{%
  \providecommand*{\frefcorname}{corollary}%
  \providecommand*{\Frefcorname}{Corollary}%
}
\frefformat{plain}{\fancyrefcorlabelprefix}{\frefcorname\fancyrefdefaultspacing#1}
\Frefformat{plain}{\fancyrefcorlabelprefix}{\Frefcorname\fancyrefdefaultspacing#1}

\newcommand*{\fancyrefdeflabelprefix}{def}
\fancyrefaddcaptions{english}{%
  \providecommand*{\frefdefname}{definition}%
  \providecommand*{\Frefdefname}{Definition}%
}
\frefformat{plain}{\fancyrefdeflabelprefix}{\frefdefname\fancyrefdefaultspacing#1}
\Frefformat{plain}{\fancyrefdeflabelprefix}{\Frefdefname\fancyrefdefaultspacing#1}

\newcommand*{\fancyrefremlabelprefix}{rem}
\fancyrefaddcaptions{english}{%
	\providecommand*{\frefremname}{remark}%
	\providecommand*{\Frefremname}{Remark}%
}
\frefformat{plain}{\fancyrefremlabelprefix}{\frefremname\fancyrefdefaultspacing#1}
\Frefformat{plain}{\fancyrefremlabelprefix}{\Frefremname\fancyrefdefaultspacing#1}

\title{
Generalized Positive Energy Representations of the Group of Compactly Supported Diffeomorphisms}
\date{\today}
\author{Bas Janssens, Milan Niestijl}

\begin{document}
	\pagenumbering{roman}
	\maketitle
	
	\begin{abstract}
		\noindent
		Motivated by asymptotic symmetry groups in general relativity, 
		we consider projective unitary representations $\overline{\rho}$ of the Lie group $\Diff_c(M)$ of compactly supported diffeomorphisms of a smooth manifold $M$ that satisfy a so-called generalized positive energy condition. In particular, this captures representations that are in a suitable sense compatible with a KMS state on the von Neumann algebra generated by $\overline{\rho}$. We show that if $M$ is connected and $\dim(M) > 1$, then any such representation is necessarily trivial on the identity component $\Diff_c(M)_0$. As an intermediate step towards this result, we determine the continuous second Lie algebra cohomology $H^2_\ct(\X_c(M), \R)$ of the Lie algebra of compactly supported vector fields. This is subtly different from Gelfand--Fuks cohomology in view of the compact support condition.
	\end{abstract}
	
	\section{Introduction}
	\pagenumbering{arabic}
	
	The mathematical results in this paper are motivated by asymptotic symmetry groups in general relativity, as in the seminal work of and Bondi, van der Burg, 		
	Metzner and Sachs \cite{Bondi_Metzner_1962, Sachs1962PR, Sachs1962}. These groups can typically be described as follows. First, one carefully selects a set of 
	boundary conditions for the gravitational fields, usually in terms of fall-off conditions at null or spacelike infinity \cite{Penrose_conformal_notes_1964, 
	PenrosePRSL1965, AshtekarHansen1978}. This gives rise to a \emph{weak asymptotic symmetry group} $G \subseteq \Diff(M)$,
	whose action on the gravitational fields preserves the specified boundary conditions. 
	In the context of classical field theory, infinitesimal symmetries yield conserved currents by Noether's theorem \cite{BarnichBrandt2002}. If one selects a 
	normal subgroup $N \subseteq G$ for which the corresponding currents are trivial
	(the group of \emph{trivial gauge transformations}),
	then the quotient $G/N$ can be interpreted as an \emph{asymptotic symmetry group}.
	For asymptotically flat space--times, this is the Bondi--Metzner--Sachs group (or BMS group for short), which is  
	a semidirect product of the Lorentz group with an infinite dimensional abelian Lie group \cite{Ashtekar_1981, Schmeding_BMS}.
	However, in general the precise form of $G/N$ depends quite sensitively on the choice of boundary conditions \cite{BrownHenneaux1986}.\\

	\noindent 
	In the present paper, we take a complementary approach. For a given group $G \subseteq \Diff(M)$ of `weak asymptotic symmetries', we expect a putative quantum theory of gravity on a space--time manifold $M$ to come with a Hilbert space $\mc{H}$ of states, and with a projective unitary representation $(\overline{\rho}, \mc{H})$
	of $G$. We investigate representations of $G$ that are subject to a (generalized) positive energy condition, whose precise meaning will be given shortly. For $M$ connected with $\dim(M) > 1$, our main result implies that the common kernel $N$ of all such representations contains the identity component $\Diff_c(M)_0$ of the group of compactly supported diffeomorphisms, assuming of course that $G$ contains $\Diff_c(M)_0$ in the first place. We therefore arrive at the conclusion that the 
	connected group $(G/N)_0$ is localized at infinity from representation theory alone, without any reference to classical field theory.\\
	
	\noindent
	In future work, we intend to isolate groups $G$ that admit non-trivial (generalized) positive energy representations, and subsequently to classify these representations. For now, let us remark that the universal central extension of the BMS group in three dimensions exhibits coadjoint orbits with energy bounded from below \cite{BarnichOblak2014, BarnichOblak2015}, hinting at the existence of induced unitary representations that are of positive energy. We refer to \cite{McCarthy_BMS_I, McCarthy_BMS_II, McCarthy_BMS_III, McCarthy_BMS_proj, Piard_BMS_1977, Piard_mackey1977} for unitary representations of the BMS
	group in four dimensions, and to \cite{BasNeeb_PE_reps_I} for an approach to positive energy representations of gauge groups in the setting of Yang-Mills theory that is similar to the one in the present paper.\\
	
\subsubsection*{Generalized positive energy representations}	
	Let us describe in more detail the type of representations that we wish to consider in the present paper. Let $\upsilon$ be a complete vector field on $M$. Then the action of $\R$ on $\Diff_c(M)$ by conjugation with the flow 
	$\Phi^{\upsilon}_t$	of $\upsilon$ gives rise to the semidirect product $\Diff_c(M) \rtimes_\upsilon \R$, which is a locally convex Lie group by \cite{Glockner_semidirect_prod_Diffc}. If $G$ is a group of diffeomorphisms that contains  both $\Diff_c(M)$ and $\{\Phi^{\upsilon}_t\,;\, t\in \R\}$, then 
	every projective unitary representation of $G$ pulls back to a projective unitary representation of $\Diff_c(M) \rtimes_{\upsilon}\R$.
	In order to show that the kernel of the $G$-representation contains $\Diff_{c}(M)_0$, we can therefore restrict attention to the group
	$\Diff_c(M) \rtimes_{\upsilon}\R$, and we will do so from now on.\\
	
	\noindent We consider projective unitary representations $\overline{\rho} \colon G \rightarrow \PU(\mc{H})$ of $G = \Diff_c(M) \rtimes_{\upsilon}\R$
	that are \emph{smooth}, in the sense that they possess a dense set $\mc{H}^{\infty} \subseteq \mc{H}$
	of smooth vectors. Such a representation is said to be of \emph{positive energy} at $\upsilon$ if 
	the strongly continuous one-parameter group $[U_t] := \overline{\rho}(\id_{M}, t)$ of projective unitary operators has a generator 
	$H := -i\restr{\frac{d}{dt}}{t=0} U_t$ (defined up to an additive constant) whose spectrum is bounded from below.\\
	
	\noindent 
	If $\upsilon$ admits an interpretation as a future timelike vector field, then $H$ is the corresponding Hamilton operator, and 
	the positive energy condition is quite natural from a physical perspective. Positive energy representations of possibly infinite-dimensional Lie groups have 
	consequently been the subject of a great deal of research \cite{Segal_unreps_of_some_inf_dim_gps, Segal_Loop_Groups, Wassermann_fusion_PE_reps, Toledano_PE_reps_non_simply_ctd, neeb_semibounded_inv_cones, Tanimoto_ground_state_reps, Neeb_semibounded_hilbert_loop, Neeb_bdd_semibdd_reps, 
	neeb_russo_ground_state_top, BasNeeb_PE_reps_I}. For $\Diff(S^1)$, they were considered in e.g.\
	\cite{Segal_unreps_of_some_inf_dim_gps, Segal_Loop_Groups,
	GoodmannWallach_projrepr_diff, Neeb_Salmasian_Virasoro_pe}. \\
	 	 
	\noindent
	In order to describe systems at positive temperature, the positive energy condition can be replaced by the \emph{KMS condition}.
	KMS states (for Kubo--Martin--Schwinger) on operator algebras were introduced by Hugenholtz, Haag and Winnink \cite{HaagHugenholtzWinnink1967} 
	in order to describe quantum statistical systems at positive temperature.
	In the context of projective unitary representations, 
	we say that $\overline{\rho}$ is \textit{KMS at $\upsilon$ relative to $\Diff_c(M)$} if there exists a normal state $\phi$ on the von Neumann algebra $\mc{N} := \rho(\Diff_c(M))^{\prime \prime}$ that satisfies the KMS condition w.r.t.\ the automorphism group 
	$\sigma : \R \to \Aut(\mc{N}), \; {t \mapsto \Ad(\overline{\rho}(\id_M, t))}$. If for such a state $\phi$, the canonical cyclic vector $\Omega_\phi$ in the corresponding GNS Hilbert space $\mc{H}_\phi$ defines a smooth ray for the associated projective unitary representation ${\overline{\rho}_\phi : \Diff_c(M) \to \PU(\mc{H}_\phi)}$, we say that $\overline{\rho}$ is \textit{smoothly-KMS}  (cf.\ \Fref{def: kms_rep} and \cite[Lem.\ 5.8]{Milan_reps_jets}).\\
	
	\noindent
	The notion of a KMS state on a von Neumann algebra is closely related to Tomita--Takesaki modular theory \cite[Ch.\ VIII]{Takesaki_II}.
	It plays an important role in the operator algebraic formulation of quantum statistical 
	physics \cite{HaagHugenholtzWinnink1967}, \cite[Ch.~5.3]{bratelli_robinson_2}, and in algebraic quantum field theory 
	\cite{Borchers_mod_thermal_states, Borchers_QFT_Tomita, BuchholzSummers_geom_mod_action}, 
	motivating our desire to consider representations that are suitably compatible with a KMS state.\\
	
	\noindent
	In order to handle positive energy representations and KMS representations at the same time, 
	we study \emph{generalized positive energy representations}, a notion that is sufficiently flexible to capture 
	positive energy representations as well as a large class of KMS representations \cite{Milan_reps_jets, MilanPhDThesis}.
	We say that $(\overline{\rho}, \mc{H})$ is of \textit{generalized positive energy} if
	there exists a dense subspace $\mc{D} \subseteq \mc{H}^{\infty}$ of smooth vectors such that for every $\psi \in \mc{D}$,
	 the expected energy
	$$ \mu : \mrm{P}(\mc{H}^\infty) \to \R, \qquad \mu([\psi]) := \frac{1}{\|\psi\|^2} \langle \psi, H \psi \rangle  $$
	is bounded below on the $\Diff_c(M)_0$-orbit $\mc{O}_{[\psi]} \subseteq \mathrm{P}(\mc{H})$	(cf.\ \Fref{def: qpe}).\\ 

\noindent
	It is important to mention that the (generalized) positive energy condition is invariant under the adjoint action of $\Diff_c(M)$ on $\X_c(M) \rtimes_\upsilon \R$, 
	in the sense that $\overline{\rho}$ is of (generalized) positive energy at $\upsilon$ if and only if it is of (generalized) positive energy at 
	$\Ad_f(\upsilon) := {T(f) \circ \upsilon \circ f^{-1}}$ for all $f \in \Diff_c(M)$.
	More generally, if $\overline{\rho}$ extends to a projective representation of a 
	Lie group $G$ that contains both $\Diff_{c}(M)$ and the flow of $\upsilon$,
	then the choice of $\upsilon$ is only significant up to the adjoint action of $G$.
	We expect to find interesting (generalized) positive energy representations only when the adjoint orbit of $G$ through $\upsilon$ generates a convex cone that is pointed. This is the case, for example, in the context of the BMS group for a suitable choice of $\upsilon$.

	\subsubsection*{The main result and its consequences}
	\noindent
	The following is our main result.
	
	\begin{theoremA}
		Suppose that $M$ is connected and that $\dim(M) > 1$. Let $\upsilon \in \X(M)\setminus \{0\}$ be a complete vector field on $M$. Let $\overline{\rho} : \Diff_c(M) \rtimes_\upsilon \R \to \PU(\mc{H})$ be a smooth projective unitary representation that is of generalized positive energy at $\upsilon$. Then $\Diff_c(M)_0 \subseteq \ker \overline{\rho}$.
	\end{theoremA}

	\noindent
	This has at least three noteworthy consequences. First of all, it follows immediately that $\Diff_c(M)_0$ is in the kernel of every \emph{positive energy representation}. 
	Secondly, $\Diff_c(M)_0$ is in the kernel of every \emph{smoothly-KMS representation} whose image generates a factor in $\mathrm{B}(\mc{H})$  (Corollary~\ref{cor: KMS_reps_trivial}).
	And finally,  $\Diff_c(M)_0$ is in the kernel of every projective unitary representation $(\overline{\rho},\mc{H})$ of $\Diff_{c}(M)$ that is \emph{bounded}, 
	i.e.\ continuous in the norm topology (Corollary~\ref{cor: bdd_reps_trivial}).
	
	\subsubsection*{Compactly supported Lie algebra cohomology}

	\noindent
	From a technical point of view, a key step towards Theorem~A is determining the continuous second Lie 
	algebra cohomology $H^2_\ct(\X_c(M), \R)$ of the Lie algebra $\X_c(M)$ of compactly supported vector fields, equipped with its natural locally convex LF-topology.
	Indeed, any smooth projective unitary representation $(\overline{\rho},\mc{H})$ of $\Diff_c(M) \rtimes_\upsilon \R$ gives rise to a canonical class in 
	$H^2_\ct(\X_c(M), \R)$ that controls the corresponding central extension at the infinitesimal level \cite{BasNeeb_ProjReps}. 
	If $\overline{\rho}$ is of generalized positive energy, then this cohomology class carries substantial information about the 
	kernel of $\overline{\rho}$ (\Fref{prop: generalized_pe_cauchy_schwarz}).
	Our main result in this regard is the following theorem:
	
	\newpage
	\begin{theoremB}
		Let $M$ be a smooth manifold.
		\begin{enumerate}
			\item If $\dim(M) > 1$, then $H^2_\ct(\X_c(M), \R) = \{0\}$.
			\item If $\dim(M) = 1$, then $H^2_\ct(\X_c(M),\R) \cong H^0_{\mrm{dR}}(M)$ is the de Rham cohomology of $M$ in degree 0.
		\end{enumerate}
		\end{theoremB}
	
	\noindent
	The continuous Lie algebra cohomology $H^n_\ct(\X_c(M), \R)$ is closely related to the \emph{Gelfand--Fuks cohomology} $H^n_\ct(\X(M), \R)$, 
	the continuous Lie algebra cohomology of the Fr\'echet--Lie algebra $\X(M)$ of all vector fields. 
	However, these two notions do \emph{not} in general coincide \cite{Shnider_cohom_cpt_vfields}, and many of the tools that are used in Gelfand--Fuks
	 cohomology (such as Bott's homotopy operators \cite{Bott_notes_GF_cohom}) break down in the compactly supported case.\\
	
	\noindent
	The study of Gelfand--Fuks cohomology was initiated by Gelfand and Fuks
	\cite{GelfandFuks_cohom_circle, GelfandFuks_cohomTangential_I, GelfandFuks_cohomTangential_II, GelfandFuks_formal_cohom},
	and carried further by Bott, G.\ Segal, Haefliger and many others 
	\cite{Guillemin_cohom_vfields, Haefliger_cohom_vfields, BottSegal_cohom_vfields, Tsujishita_book, Losik_diag_cohom_vfields}, 
	see also  \cite{Bott_notes_GF_cohom, Fuks_book} and the recent exposition \cite{Lukas_GF_cohom}.
	For compact manifolds $M$ with $\dim(M) > 1$, it is well known that 
	the second Gelfand--Fuks cohomology $H^2_\ct(\X(M), \R)$ vanishes (\cite[Thm.\ 4.13 and Cor.\ 4.25]{Lukas_GF_cohom}), 
	covering part 1 of Theorem~B in the compact case.\\
	
	\noindent 
	Nevertheless, Theorem~B appears to be new if $M$ is a non-compact manifold, which is of course the case of primary interest in the context of 
	asymptotic symmetry groups.
	The proof is very much inspired by \cite{BasCornelia_ce_ham}, and by joint work in progress \cite{BasCorneliaLeo_CEdiv} with Cornelia Vizman and Leonid Ryvkin on the second 
	Lie algebra cohomology of the Lie algebra of exact volume preserving vector fields. We are grateful for their kind permission to use ideas from this 
	unpublished work in the current setting.\\
	
	\subsubsection*{Outline and further references}
	
	\noindent
	The paper is organized as follows. In \Fref{sec: preliminaries} we recall various preliminary definitions and observations.
	In \Fref{sec: gen_pe_reps_vfields} we proceed with the proof of Theorem~A and its consequences, 
	subject to the assumption that $H^2_\ct(\X_c(M), \R) = 0$ if $\dim(M) > 1$.
	Finally, in \Fref{sec: second_cts_LA_cohom}, we justify this assumption by 
	determining the continuous second Lie algebra cohomology $H^2_\ct(\X_c(M), \R)$ for arbitrary manifolds $M$, 
	culminating in Theorem~B. \\
	
	\noindent
	KMS representations of infinite dimensional Lie groups were studied by Str\v{a}til\v{a} and Voiculescu for $\U(\infty)$ \cite{StratilaVoiculescu1978}, see \cite{Boyer1980, NeebRusso2024} for related results 
	in the context of Hilbert--Lie groups.
	For finite dimensional Lie groups, 
	KMS representations that generate a factor of type~$\mrm{I}$ were fully classified in \cite{Tobias_typeI_factor_reps}.
	Projective unitary KMS representations were constructed for the loop group $C^\infty(S^1, \U(N))$ in \cite{Carey_Hannabuss_KMS_1987, Buchholz_current_KMS}, and for certain Sobolev maps from $\R$ to $\U(N)$ in \cite{CareyHannabuss_KMS_1992}. 
	Various other examples of KMS representations are given in \cite[Sec.\ 5.2.2]{Milan_reps_jets}.

	\subsubsection*{Acknowledgments}
	This research is supported by the NWO grant 639.032.734 \dquotes{Cohomology and representation theory of infinite-dimensional Lie groups}. The authors are also grateful to Lukas Miaskiwskyi, Cornelia Vizman and Leonid Ryvkin for illuminating discussions on the topic, and to Cornelia Vizman and Leonid Ryvkin 
	for their permission to use ideas from our unpublished work \cite{BasCorneliaLeo_CEdiv}.
	We would like to thank the anonymous referees for suggestions that have substantially improved the structure of the paper.
	\section{Preliminaries}\label{sec: preliminaries}
	
	We briefly discuss projective unitary representations for locally convex Lie groups, and recall 
	some properties of generalized positive energy representations that we will need in \Fref{sec: gen_pe_reps_vfields}.
		
	\subsection{Projective unitary representations}
	
	\noindent
	Let $G$ be a locally convex Lie group, in the sense of Bastiani \cite{bastiani, milnor_inf_lie, neeb_towards_lie}, with Lie algebra~$\g$.\\
	
	\noindent
	Let $\mcD$ be a complex pre-Hilbert space with Hilbert space completion $\mc{H}$. We denote by $\mc{L}(\mcD)$ the set of linear operators $\mcD\to \mcD$. Define the algebra
	$$ \mc{L}^\dagger(\mcD) := \SET{X \in \mc{L}(\mcD) \st \exists X^\dagger \in \mc{L}(\mcD) \st \forall \psi, \eta \in \mcD \st \langle X^\dagger \psi, \eta\rangle = \langle \psi, X\eta\rangle}.$$
	It carries a natural involution $T \mapsto T^\dagger$. Let $I$ denote the identity on $\mcD$ and define the Lie algebra
	\[\mf{u}(\mcD) := \SET{X \in \mc{L}^\dagger(\mcD) \st X^\dagger + X = 0}.\]
	Define also the Lie algebra $\pu(\mcD) := \mf{u}(\mcD) / i \R I$. 
	
	\begin{definition}We define smooth projective unitary representations as follows:
		\begin{itemize}
			\item A unitary representation $(\rho, \mc{H})$ of $G$ is \emph{continuous} if $g \mapsto \rho(g)\psi$ is continuous for 
			every $\psi \in \mc{H}$. Similarly,  a projective unitary representation $(\overline{\rho}, \mc{H})$ is continuous 
			if the orbit map $g \mapsto \overline{\rho}(g)[\psi]$ is continuous for every $\psi \in \mc{H} \setminus\{0\}$.
			\item If $(\rho, \mc{H})$ is a unitary representation of $G$, then a vector $\psi \in \mc{H}$ is called \textit{smooth} if the orbit map $G \to \mc{H}, g \mapsto \rho(g)\psi$ is smooth. We denote by $\mc{H}^\infty$ the set of smooth vectors in $\mc{H}$, and we call $\rho$ smooth if $\mc{H}^\infty$ is dense in $\mc{H}$. 
			\item Similarly, if $(\overline{\rho}, \mc{H})$ is a projective unitary representation of $G$, then a ray $[\psi] \in \mrm{P}(\mc{H})$ is said to be \textit{smooth} if the orbit map $G \to \mrm{P}(\mc{H}), g \mapsto \overline{\rho}(g)[\psi]$ is smooth. We denote by $\mrm{P}(\mc{H})^\infty$ the set of smooth rays, and we say that $\overline{\rho}$ is smooth if $\mrm{P}(\mc{H})^\infty$ is dense in $\mrm{P}(\mc{H})$.
			\item A unitary representation of a locally convex Lie algebra $\g$ on $\mcD$ is a homomorphism $\pi : \g \to \mf{u}(\mcD)$ of Lie algebras. It is called continuous if the map $\xi \mapsto \pi(\xi)\psi$ is continuous for every $\psi \in \mc{D}$. A projective unitary representation of $\g$ on $\mc{D}$ is a Lie algebra homomorphism $\olpi : \g \to \pu(\mcD)$.
		\end{itemize}
	\end{definition}
	
	\begin{remark}
		A smooth unitary representation $(\rho, \mc{H})$ of $G$ defines a unitary $\g$-representation $d\rho : \g \to \mf{u}(\mc{H}^\infty)$ on $\mc{H}^\infty$ by $d\rho(\xi)\psi := \restr{\frac{d}{dt}}{t=0}\rho(\gamma_t)\psi$, where $\gamma : \R \to G$ is a smooth curve satisfying $T_0(\gamma) = \xi$. If $G$ is finite-dimensional, then $\mc{H}^\infty$ is dense in $\mc{H}$ for any continuous unitary representation $\rho$ of $G$, by a result of G\aa rding \cite{Garding_domain1947} (cf.\ \cite[Prop.\ 4.4.1.1]{Warner_book_1}). The analogous statement is generally false for infinite-dimensional Lie groups \cite{Beltita_Neeb_nonsmooth_repr}.
	\end{remark}
	
	\begin{definition} A \emph{central extension} of $G$ by the circle group $\T$ is an exact sequence 
	\begin{equation}\label{eq:Centralegroepsuitbreiding}
		1 \rightarrow \T \rightarrow \circled{G} \rightarrow G \rightarrow 1
	\end{equation}
	of groups for which the image of $\T$ in $\circled{G}$ is central. It is a central extension of \emph{topological groups} if 
	$\circled{G}$ is a topological group, and 
	the group homomorphisms in \eqref{eq:Centralegroepsuitbreiding} are continuous. It is a central extension of \emph{locally convex Lie groups} if 
	$\circled{G}$ is a locally convex Lie group,
	the group homomorphisms in \eqref{eq:Centralegroepsuitbreiding} are 
	smooth, and $\circled{G} \rightarrow G$ is a locally trivial smooth principal $\T$-bundle. An \emph{isomorphism} $\circled{G} \rightarrow \circled{G}'$ 
	of central extensions is an isomorphism of groups (topological groups, Lie groups) that induces the identity on $G$ and~$\T$.
	\end{definition}
	
	\begin{definition}
	A \emph{central extension of Lie algebras} is an exact sequence
	\begin{equation}\label{eq:CentraleAlgebraUitbreiding}
		0 \rightarrow \R \rightarrow \circled{\g} \rightarrow \g \rightarrow 0
	\end{equation}
	of Lie algebras for which the image of $\R$ in $\circled{\g}$ is central. It is a \emph{continuous} central extension
	if $\circled{\g}$ is a locally convex Lie algebra, and the Lie algebra homomorphisms in \eqref{eq:CentraleAlgebraUitbreiding}
	are continuous. 
	\end{definition}

	\noindent
	Any central $\T$-extension of locally convex Lie groups determines a corresponding continuous central $\R$-extension of Lie algebras.\\
	
	\noindent
	If $\overline{\rho} : G \to \PU(\mc{H})$ is a continuous projective unitary representation of $G$, 
	then the pullback
	\begin{equation}\label{eq: circledG}
		\circled{G} := \big\{(g, U) \in G \times \U(\mc{H}) \, ; \, \overline{\rho}(g) = [U]\big\}	
	\end{equation}
	is a central extension $\circled{G} \rightarrow G$ of topological groups, and $\overline{\rho}$ lifts to a continuous unitary representation 
	$\rho \colon {\circled{G} \rightarrow \U(\mc{H})}$ that satisfies $\rho(z) = zI$ for all $z \in \T$. We call $\rho$ the \textit{lift} of $\overline{\rho}$.\\ 

	\noindent If, moreover, the projective unitary representation $\overline{\rho}$ of $G$ is smooth, then $\circled{G}$ is a locally convex Lie group,  
	and $\circled{G} \rightarrow G$ is a central extension of locally convex Lie groups
	\cite[Thm.\ 4.3]{BasNeeb_ProjReps}. The lift $\rho$ is smooth \cite[Cor.\ 4.5]{BasNeeb_ProjReps}, and $\mrm{P}(\mc{H})^\infty = \mrm{P}(\mc{H}^\infty)$ \cite[Thm.\ 4.3]{BasNeeb_ProjReps}. Suppose that $(\overline{\rho}_1, \mc{H}_1)$ and $(\overline{\rho}_2, \mc{H}_2)$ are two smooth projective unitary representations with lifts $\rho_1 : \circled{G}_1 \to \U(\mc{H}_{1})$ and $\rho_2 : \circled{G}_2 \to \U(\mc{H}_{2})$, respectively. Then $\overline{\rho}_1$ and $\overline{\rho}_2$ are unitarily equivalent if and only if there exists an isomorphism 
	$\Phi : {\circled{G}_1 \to \circled{G}_2}$ of central extensions and a unitary $U : \mc{H}_{1} \to \mc{H}_{2}$ such that $\rho_2(\Phi(x)) = U\rho_1(x)U^{-1}$ for all $x \in \circled{G}_1$ \cite[Thm.\ 7.3]{BasNeeb_ProjReps}.\\
	
	\noindent Analogously, any projective unitary $\g$-representation $\olpi$ on $\mcD$ can be lifted to a unitary representation $\pi : \circled{\g} \to \mf{u}(\mcD)$ of some central $\R$-extension $\circled{\g}$ of $\g$, by considering the pull-back of $\mf{u}(\mcD) \to \pu(\mcD)$ along $\olpi$.

	\begin{definition}\label{def: proj_reps_la_cts}
		We call a projective unitary representation $\olpi : \g \to \pu(\mcD)$ \textit{continuous} if its 
		lift $\pi : \circled{\g} \to \mf{u}(\mcD)$ is continuous.
	\end{definition}
	
	\subsection{Cohomology of Lie algebras and Lie groups}	
	
	Smooth projective unitary representations of $G$ give rise to central $\T$-extensions of locally convex Lie groups, and these in turn determine continuous central $\R$-extensions of the Lie algebra $\g$. The latter can be described in terms of continuous Lie algebra cohomology.

	\begin{definition}[Lie algebra cohomology]\label{def: cts_LA_cohom} 
	Let $E$ be a module over a Lie algebra $\g$.
		\begin{itemize}
		\item The Lie algebra cohomology $H^\bullet(\g, E)$ of $\g$ with values in $E$ is the cohomology of the complex 
			$C^\bullet(\g, E)$, where $C^q(\g,E)$ consists of alternating 
			multilinear maps $\g^q \rightarrow E$ for $q\geq 0$, and it is zero for $q<0$.
			The differential $d_\g : C^\bullet(\g, E) \to C^{\bullet + 1}(\g,E)$
			is given by	
			\begin{align}\label{eq: differential_la_coh}
				\begin{split}
					d_\g\omega(\xi_0, \ldots, \xi_q) := &\sum_{j=0}^q (-1)^j \xi_j\cdot \omega(\xi_0, \ldots, \widehat{\xi_j}, \ldots, \xi_q) \\
					+ &\sum_{0 \leq i < j \leq q}(-1)^{i+j}\omega([\xi_i, \xi_j], \xi_0, \ldots, \widehat{\xi_i}, \ldots, \widehat{\xi_j}, \ldots, \xi_q ).
				\end{split}
			\end{align}
			As usual, the arguments in \eqref{eq: differential_la_coh} with a caret are to be omitted. Unless mentioned otherwise, the vector space $\R$ is considered as a trivial $\g$-module.
		\item If $\g$ is a locally convex Lie algebra and $E$ a topological $\g$-module, then the
		continuous Lie algebra cohomology $H_{\ct}^\bullet(\g, E)$ is the cohomology of the subcomplex 
		$C_{\ct}^\bullet(\g, E)$ of continuous alternating multilinear maps.		
	\end{itemize}
	\end{definition}
	
	\noindent The continuous central extensions of $\g$ by $\R$ are classified up to isomorphism by $H_{\ct}^2(\g, \R)$, the continuous second Lie algebra cohomology with trivial coefficients \cite[Prop.\ 6.3]{BasNeeb_ProjReps}. 
	In order to study smooth projective unitary representations of $G = \Diff_{c}(M)$, it is sensible to determine
	$H^2_{\ct}(\g,\R)$ for the Lie algebra $\g = \X_{c}(M)$. This is done in Section~\ref{sec: second_cts_LA_cohom}.
	\\

	\noindent
	The interpretation of the second Lie algebra cohomology in terms of central extensions is already implicit in the work of 
	Schur \cite{Schur1904, Schur1907}, and its use in the projective representation theory of Lie groups
	was pioneered by E.~Wigner \cite{EugeneWigner1939} and V.~Bargmann \cite{Bargmann1954}, see \cite{TuynmanWiegerinck1987}
	for an exposition and further references.\\
	
	\noindent
	To some extent, Lie algebra cohomology functions as an infinitesimal counterpart of Lie group cohomology, 
	with their relation typically given by a van Est-type spectral sequence \cite{WimVanEst1958}.
	Although group cohomology for discrete groups admits a good description in terms of
	commutative algebra \cite{Brown1982}, the appropriate cohomology theory in the context of Lie groups requires 
	a bit more care. There are in fact many different flavours of group cohomology for Lie groups, 
	grounded either in \v{C}ech cohomology or in explicit cocycle models \cite{Segal70, DavidWigner1973, Brylinski00}.
	With the notable exception of 
	bounded cohomology \cite{Monod2001}, they mostly agree on the domain for which they are intended, 
	see \cite{WagemannWockel2015} for an overview and comparison, as well as for further references.

	\subsection{Generalized positive energy (GPE) representations}
	
	\noindent
	In the following, $G$ denotes a locally convex Lie group which is regular in the sense of Milnor \cite[Def.\ 7.6]{milnor_inf_lie} (cf.\ \cite[Def.\ II.5.2]{neeb_towards_lie}). We denote by $\g$ the Lie algebra of $G$.\\

	\noindent
	This paper is concerned with projective unitary representations $\overline{\rho}$ of $G$ that satisfy a so-called generalized positive energy condition. This class of representations was introduced in \cite[Sec.\ 4]{Milan_reps_jets}. It includes representations that satisfy a positive energy condition, and also representations that are in a suitable sense compatible with a KMS state on the von Neumann algebra generated by $\overline{\rho}(G)$.\\
	
	\noindent
	We first introduce the precise definitions, and then review the restrictions that the cohomology class $[\omega] \in H^2_\ct(\g, \R)$ associated to such a representation imposes on its kernel.
	These restrictions play a crucial role in the proof of Theorem~A.
	
	\begin{definition}\label{def: qpe}
		Let $\mc{D}$ be a complex pre-Hilbert space with Hilbert space completion $\mc{H}$. Let $\h$ be a locally convex topological Lie algebra.
		\begin{itemize}
			\item Let $\pi : \h \to \mf{u}(\mcD)$ be a continuous unitary representation of $\h$ on $\mcD$. We say that $\pi$ is of \textit{positive energy} at $\xi \in \h$ if 
			$$ \inf_{\psi \in \mcD}\langle \psi, -i\pi(\xi)\psi\rangle \geq 0.$$
			We say that $\pi : \h \to \mf{u}(\mcD)$ is of \textit{generalized positive energy} (GPE) at $\xi \in \h$ if there exists a $1$-connected regular Lie group $H$ with Lie algebra $\h$ and a dense linear subspace $\mcD_\xi \subseteq \mcD$ such that
			\begin{equation}\label{eq: qpe_la}
				\forall \psi \in \mcD_\xi \st \inf_{h \in H}\langle \psi, -i\pi(\Ad_h(\xi)) \psi\rangle > - \infty.
			\end{equation}
			\item Let $\olpi$ be a continuous projective unitary representation of $\h$
			with lift $\pi \colon {\circled{\h} \to \mf{u}(\mcD)}$. We say that $\olpi$ is of (generalized) positive energy at $\xi \in \h$ if $\pi$ is so at some $\circled{\xi} \in \circled{\h}$ covering $\xi$.
			\item A smooth unitary representation $(\rho, \mc{H})$ of $G$ is of (generalized) positive energy at $\xi \in \g$ if the derived representation $d\rho$ of $\g$ on $\mc{H}^\infty$ is so.
			\item Let $\overline{\rho}$ be a smooth projective unitary representation of $G$ on $\mc{H}$ with lift $\rho \colon {\circled{G} \to \U(\mc{H})}$. We say that $\overline{\rho}$ is of (generalized) positive energy at $\xi \in \g$ if there exists an element $\circled{\xi} \in \circled{\g}$ covering $\xi$ such that $\rho$ is of (generalized) positive energy at $\circled{\xi}$.\\
		\end{itemize}
	\end{definition}
	
	\begin{remark}\label{rem:RegularCase}
	Suppose that $\pi$ is a continuous unitary representation of $\g$ on $\mc{D}$ that is of GPE at some element in $\g$. Then the group $H$ in \eqref{eq: qpe_la} is the simply connected cover $\widetilde{G}_0$ of the identity component of $G$, because two regular $1$-connected Lie groups are isomorphic if their Lie algebras are so \cite[Cor.\ 8.2]{milnor_inf_lie}, and $\widetilde{G}_0$ is regular whenever $G$ is so \cite[Thm.\ V.1.8]{neeb_towards_lie}.\\
	\end{remark}

	\begin{remark}\label{rem: pe_implies_gpe}
		If a unitary representation $(\rho, \mc{H})$ of $G$ is of positive energy at $\xi \in \g$, then it is also of 
		generalized positive energy 
		at $\xi$. 
		Indeed, since 
		\[\langle \psi, d\rho(\Ad_g^{-1}(\xi))\psi\rangle = \langle \rho(g)\psi, d\rho(\xi)\rho(g)\psi\rangle \qquad \text{ for all } g \in G, \; \xi \in \g \text{ and }\psi \in \mc{H}^\infty,\]
		\Fref{rem:RegularCase} implies that $\rho$ is of generalized positive energy at $\xi \in \g$ if and only if 
		\begin{equation}\label{eq:VergelijkingVoorGPEKegel}
		I(\xi, \psi) := \inf_{g_0 \in G_0} \langle \rho(g_0)\psi, -id\rho(\xi) \rho(g_0)\psi\rangle > - \infty
		\end{equation}
		for all $\psi$ in some linear subspace $\mcD_\xi \subseteq \mc{H}^\infty$ that is dense in $\mc{H}$. If $\rho$ is of positive energy at $\xi$, then the left hand side of \eqref{eq:VergelijkingVoorGPEKegel} is nonnegative for any $\psi \in \mc{H}^\infty$, since $\mc{H}^{\infty}$ is $G$-invariant.
	\end{remark}
	
	\begin{remark}\label{rem: gpe_and_convex_cones}
		Let $(\rho, \mc{H})$ be a smooth unitary representation of $G$.
		\begin{itemize}
			\item
			The \emph{generalized positive energy cone}
			\[\mc{C}(\rho) := \SET{\xi \in \g \st \rho \text{ is of GPE at } \xi } \]
			is $\Ad_G$-invariant for the (not necessarily connected) Lie group $G$. Indeed, if $\xi \in \mc{C}(\rho)$ and $g\in G$, then 
			\eqref{eq:VergelijkingVoorGPEKegel} for $\xi$ and $\psi \in \mc{D}_{\xi}$ implies the corresponding inequality for 
			$\xi' = \Ad_g(\xi)$ and $\psi' \in \mc{D}_{\xi'}$ with $\mc{D}_{\xi'} := \rho(g)\mc{D}_\xi$. 
			\item 
			If $\xi \in \mc{C}(\rho)$, then $\mc{C}(\rho)$ 
			also contains the $\Ad_{G_0}$-invariant convex cone generated by $\xi$. To see this, suppose 
			that \eqref{eq:VergelijkingVoorGPEKegel} is satisfied for $\psi \in \mc{D}_\xi \subseteq \mc{H}^\infty$. 
			Let $n \in \N$, $g_k \in G_0$ and $c_k \geq 0$ for $k \in \{1, \ldots, n\}$, and define $C := \sum_{k=1}^n c_k$. Then for 
			$\xi' := \sum_{k=1}^n c_k \Ad_{g_k}(\xi)$ we have that $I(\xi', \psi) \geq  C\cdot I(\xi, \psi)$, so that $\xi'\in \mc{C}(\rho)$. 
			\item
			If, moreover, $\rho$ is of positive energy at $\xi\in \g$, then it is of positive energy at every element of the 
			$\Ad_G$-invariant closed convex cone $\mc{C}_\xi\subseteq \g$ generated by $\xi$. 
			It follows that $\ker(d\rho)$ contains the $\Ad_G$-invariant closed 
			ideal $\mc{C}_\xi \cap -\mc{C}_\xi$ of $\g$. If $\mc{C}_\xi \cap - \mc{C}_\xi = \g$, 
			then \cite[Prop.\ 3.4]{BasNeeb_ProjReps} implies that $G_0 \subseteq \ker \rho$. 
			\item
			In particular, if $\g$ admits no non-zero proper $\Ad_G$-invariant closed ideals, then $G_0$ can only act non-trivially in a smooth unitary representation of $G$ that is of positive energy at $\xi$ if the cone $\mc{C}_\xi$ is pointed. 
			This is the situation for $G = \Diff_c(M)$ if $M$ is connected, as we shall see in \Fref{cor: no_stable_ideals}. \\
		\end{itemize}	
	\end{remark}

	\begin{remark}
		A smooth projective unitary $G$-representation $\overline{\rho}$ with lift $\rho$ is of GPE at $\xi \in \g$ if and only if for some (and hence any) $\circled{\xi} \in \circled{\g}$ covering $\xi$, the function
		$$ \mu : \mrm{P}(\mc{H}^\infty) \to \R, \; \mu([\psi]) := \frac{1}{\|\psi\|^2} \langle \psi, -id\rho(\circled{\xi}) \psi \rangle $$
		is bounded below on the $G_0$-orbit $\mc{O}_{[\psi]} \subseteq \mrm{P}(\mc{H}^\infty)$ for all $\psi$ in some dense linear subspace $\mcD_\xi \subseteq \mc{H}^\infty$.\\
	\end{remark}
	
	\noindent
	The following observation plays a crucial role in \Fref{sec: gen_pe_reps_vfields}. It shows that the cohomology class in $H^2_{\mathrm{ct}}(\g,\R)$ associated 
	to a projective GPE-representation carries substantial information about the kernel.

	\begin{proposition}[{\cite[Prop.\ 4.4]{Milan_reps_jets}}]\label{prop: generalized_pe_cauchy_schwarz}
		Let $\overline{\pi} : \g \to \pu(\mcD)$ be a projective unitary representation of $\g$ that is of generalized positive energy at $\xi \in \g$. Let $\omega \in C_{\ct}^2(\g, \R)$ be a 2-cocycle whose class $[\omega]$ in  $H^2_\ct(\g, \R)$ corresponds to the central 
		extension $\circled{\g} \rightarrow \g$. Suppose that $\eta \in \g$ satisfies $[[\xi, \eta], \eta] = 0$. Then $\omega([\xi, \eta], \eta) \geq 0$ and
		$$ \omega([\xi, \eta], \eta) = 0 \iff \olpi([\xi, \eta]) = 0. $$
		In particular, if $[\omega] = 0$ in $H^2_\ct(\g, \R)$, then 
		$$ [[\xi, \eta], \eta] = 0 \implies \olpi([\xi, \eta]) = 0. $$
	\end{proposition}
	
	\noindent
	Before defining the notion of a KMS-representation, we briefly recall the definition of a KMS state on a von Neumann algebra (cf.\ \cite[Ch.\ 5.3]{bratelli_robinson_2}, \cite[Ch.\ VIII]{Takesaki_II}). Let $\St := \SET{ z \in \C \st 0 < \mrm{Im}(z) < 1}$.

	\begin{definition}\label{def: mod_condition_kms}
		Let $\phi$ be a normal state on a von Neumann algebra $\mc{N}$, and let $\sigma \colon {\R \to \Aut(\mc{N})}$ be a one-parameter group.
		\begin{itemize}
			\item We say that $\phi$ satisfies the \textit{modular condition} for $\sigma$ if:
			\begin{enumerate}
				\item $\phi = \phi \circ \sigma_t$ for every $t \in \R$. 
				\item For every pair $x,y \in \mc{N}$, there is a bounded continuous function $F_{x,y} : \overline{\St} \to \C$, holomorphic on $\St$, so that for every $t \in \R$ we have:
				\begin{align*}
					F_{x,y}(t) &= \phi(\sigma_t(x)y), \\
					F_{x,y}(t+i) &= \phi(y\sigma_t(x)). 
				\end{align*}
			\end{enumerate}
			\item We say that $\phi$ is \textit{$\sigma$-KMS} if it satisfies the modular condition for the automorphism group $t \mapsto\sigma_{-t}$.
		\end{itemize}
	\end{definition}
	
	\begin{definition}\label{def: kms_rep}
	Let $\xi \in \g$, and let $N \subseteq G$ be an embedded Lie subgroup such that $e^{t\xi}Ne^{-t\xi} \subseteq N$ for all $t \in \R$. 
	\begin{enumerate}
			\item A continuous unitary representation $(\rho, \mc{H})$ of $G$ is \emph{KMS at $\xi$ relative to $N$} if the von Neumann 
			algebra $\mc{N} := \rho(N)^{\prime \prime}$ admits a normal state $\phi$ that is $\sigma$-KMS for $\sigma_t(x) := \rho(e^{t\xi})x\rho(e^{-t\xi})$. It is \emph{smoothly KMS} if additionally $n \mapsto \phi(\rho(n))$ is a smooth function $N \rightarrow \C$. 
			\item Let $\overline{\rho} \colon G \to \PU(\mc{H})$ be a smooth projective unitary representation of $G$, with lift
			$\rho \colon {\circled{G} \to \U(\mc{H})}$. Let $\circled{\g}$ be the Lie algebra of $\circled{G}$ and let $\circled{N} \subseteq \circled{G}$ be the Lie subgroup covering $N$. We say that $\overline{\rho}$ is smoothly-KMS at $\xi \in \g$ relative to $N$ if there exists $\circled{\xi}\in \circled{\g}$ covering $\xi$ such that $\rho$ is smoothly-KMS at $\circled{\xi}$ relative to $\circled{N}$.
		\end{enumerate}
	\end{definition}
	
	\noindent
	Various examples of KMS-representation are considered in \cite[Sec.\ 5.2.2]{Milan_reps_jets}.\\

	\noindent
	Suppose that $N \subseteq G$ and $\xi \in \g$ are as in \Fref{def: kms_rep}. Let $\n$ be the Lie algebra of $N$, $\rho$ a unitary representation of $G$, and let $\mc{N} := \rho(N)^{\prime \prime}$ be the von Neumann algebra generated by $\rho(N)$. Define $\alpha : \R \to \Aut(N)$ by $\alpha_t(n) = e^{t\xi}ne^{-t\xi}$ and $D \in \der(\n)$ by $D\eta := [\xi, \eta]$. Suppose that the normal state $\phi$ is KMS w.r.t.\ the automorphism group $\sigma : \R\to \Aut(\mc{N}), \; \sigma_t(x) := \rho(e^{t \xi})x\rho(e^{-t\xi})$. The GNS-construction (for Gelfand--Naimark--Segal) provides a $\ast$-representation $\pi_\phi$ of the von Neumann algebra $\mc{N}$ on the GNS-Hilbert space $\mc{H}_\phi$. We therefore obtain a unitary representation $\rho_\phi := \pi_\phi \circ \rho$ of $N$ on $\mc{H}_\phi$. Letting $\Delta_\phi$ denote the modular operator associated to $\phi$ (cf.\ \cite[Ch.\ 2.5]{bratelli_robinson_1}), the representation $\rho_\phi$ extends to $N \rtimes_\alpha \R$ by setting $\rho_\phi(n,t) := \rho_\phi(n)\Delta_\phi^{-it}$. This representation is smooth if $n \mapsto \phi(\rho(n))$ is a smooth map $N \to \C$ \cite[Lem.\ 5.10]{Milan_reps_jets}. \\
	
	\noindent
	The following relates KMS-representations to the generalized positive energy condition:
	
	\begin{theorem}[{\cite[Thm.\ 5.13]{Milan_reps_jets}}]\label{thm: kms_of_qpe}
		Let $\rho$ be a unitary representation of $G$ that is smoothly-KMS at $\xi \in \g$ relative to $N \subseteq G$,
		and let $\phi \colon \mc{N} \rightarrow \C$ be as in Definition~\ref{def: kms_rep}. Then the associated unitary representation $\rho_\phi$ of $N \rtimes_{\alpha} \R$ on the GNS-Hilbert space $\mc{H}_\phi$ is smooth and of generalized positive energy at $(0,1) \in \n \rtimes_D \R$.
	\end{theorem}
	
	\section{GPE representations of $\Diff_c(M) \rtimes_\upsilon \R$}\label{sec: gen_pe_reps_vfields}
	
	\noindent
	Let $M$ be a smooth manifold of dimension $\dim(M) > 1$. If $\upsilon \in \X(M)$ is a complete vector field on $M$ with flow $\Phi^{\upsilon} : \R \to \Diff(M)$, we write $\Diff_c(M) \rtimes_\upsilon \R$ for the semidirect product of $\Diff_c(M)$ and $\R$ relative to the smooth $\R$-action on $\Diff_c(M)$ defined by 
	$\alpha_s(f) = \Phi^{\upsilon}_s \circ f \circ \Phi^{\upsilon}_{s^{-1}}$ for $s \in \R$ and $f \in \Diff_c(M)$. The corresponding Lie algebra is $\X_c(M) \rtimes \R \upsilon$, where $\upsilon$ acts on $\X_c(M)$ by the derivation $D w := [\upsilon, w]$. \\
	
	\noindent
	In Section~\ref{sec: second_cts_LA_cohom}, we will see that $H^2_\ct(\X_c(M), \R)$ is trivial for $\dim(M) > 1$. This puts severe restrictions on the class of projective unitary representations of $\X_c(M) \rtimes \R \upsilon$ that are of generalized positive energy at $\upsilon$. The following result on Lie algebra representations is the crux of the matter.
	
	\begin{theorem}\label{thm: gpe_vfields}
		Suppose that $\dim(M) > 1$. Let $\olpi : \X_c(M)\to \pu(\mcD)$ be a continuous projective unitary representation of $\X_c(M)$ on the complex pre-Hilbert space $\mcD$. Let $\mc{C} \subseteq \X(M)$ be a cone of complete vector fields on $M$, and define the open set
		$$ U := \bigcup_{\upsilon \in \mc{C}} \SET{p \in M \st \upsilon(p) \neq 0}.$$
		Suppose that for every $\upsilon \in \mc{C}$ the representation $\olpi$ extends to a continuous projective unitary representation of $\X_c(M) \rtimes \R \upsilon$ that is of generalized positive energy at $\upsilon$. Then $\X_c(U) \subseteq \ker \olpi$.
	\end{theorem}
		
	\noindent
	In \Fref{sec: ideals}, we classify invariant ideals in $\X_c(M)$. This is an intermediate step towards the proof of \Fref{thm: gpe_vfields}, which is given in \Fref{sec: pf_gpe_vfields}. Lastly, in \Fref{sec: thmA} we use this result to derive Theorem~A, which is a group-level analogue of \Fref{thm: gpe_vfields}.

	\subsection{Ideals of the Lie algebra $\X_c(M)$}\label{sec: ideals}
	
	\noindent
	Let $M$ be a smooth manifold. For any $x \in M$, let $I_x \subseteq \X_c(M)$ denote the closed ideal of vector fields that are flat at $x$. So $v \in I_x \iff j^\infty_x(v) = 0$ for $v \in \X_c(M)$. The proof of \Fref{thm: gpe_vfields} uses \Fref{prop: ideals_locally} below.
	
	\begin{definition}
		If $J \subseteq \X_c(M)$ is an ideal, define its \textit{hull} by 
		$$ h(J) := \SET{x \in M \st v(x) = 0 \text{ for all } v \in J}.$$
	\end{definition}
	
	\begin{remark}
		The set of maximal ideals in $\X_c(M)$ is given by $\{I_x \st x\in M\}$ \cite[Thm.\ 1]{Shanks_Pursell} (cf.\ \cite[prop.\ 7.2.2]{Banyaga_book_diffeo} or \cite[Prop.\ 1]{Bas_inf_nat_bundles}). Moreover, if $x \in M$ and $J \subseteq \X_c(M)$ is an ideal, then $x \in h(J)$ if and only if $j^\infty_x(v) = 0$ for all $v \in J$. Indeed, if $x\in h(J)$, then for any $w_1, \ldots, w_m \in \X_c(M)$ and $v \in J$ we have $\mc{L}_{w_1}\cdots \mc{L}_{w_m}v \in J$, as $J$ is an ideal, and so $\big(\mc{L}_{w_1}\cdots \mc{L}_{w_m}v\big)(x) = 0$. Consequently $j^\infty_x(v) = 0$. We thus see that $h(J) = {\{x \in M \st J \subseteq I_x\}}$ corresponds to the set of maximal ideals of $\X_c(M)$ containing $J$. 
	\end{remark}
	
	\begin{proposition}\label{prop: ideals_locally}
		Let $J \subseteq \X_c(M)$ be an ideal and let $x \in M$. Then either $x \in h(J)$, or there is an open neighborhood $U \subseteq M$ of $x$ such that 
		$\X_c(U) \subseteq [J, \X_c(M)] \subseteq J$.
	\end{proposition}
	\begin{proof}
		This is immediate from the proof of \cite[Lem.\ 2.1]{Bas_inf_nat_bundles}, which does not require the ideal $J \subseteq \X_c(M)$ to be maximal.
	\end{proof}
	
	\noindent
	Although $\X_c(M)$ is not simple, the following related result does hold true:
	\begin{corollary}\label{cor: no_stable_ideals}
		Assume that $M$ is connected. Suppose that $J \subseteq \X_c(M)$ is an ideal that is stable, in the sense that $\Ad_g(J) \subseteq J$ for all $g \in \Diff_c(M)_0$. Then either $J = \X_c(M)$ or $J = \{0\}$.
	\end{corollary}
	\begin{proof}
		That $J$ is stable implies that its hull $h(J) \subseteq M$ is $\Diff_c(M)_0$-invariant. Since $M$ is connected, $\Diff_c(M)_0$ acts transitively on $M$ (cf.\ \cite[p.\ 22]{Milnor_topology}). It follows that either $h(J) = \emptyset$ or $h(J) = M$. Using a partition of unity argument, \Fref{prop: ideals_locally} implies that either $J = \X_c(M)$ or $J = \{0\}$.
	\end{proof}
	
	\begin{remark}
		Suppose that $M$ is connected. Let $\rho : \Diff_c(M) \to \PU(\mc{H})$ be a smooth projective unitary representation. Let $\overline{d\rho} : \X_c(M) \to \pu(\mc{H}^\infty)$ be its derived representation. Its kernel $J := \ker \overline{d\rho}$ is a closed ideal in $\X_c(M)$ satisfying $\Ad_g(J) \subseteq J$ for all $g \in \Diff_c(M)$. So $\overline{d\rho}$ is either trivial or injective by \Fref{cor: no_stable_ideals}.
	\end{remark}
	
	\subsection{The proof of \Fref{thm: gpe_vfields}}\label{sec: pf_gpe_vfields}
	
	\noindent
	We now proceed with the proof of \Fref{thm: gpe_vfields}. Let $n := \dim(M) >1$. We start with a lemma that concerns the local situation near a regular point of a vector field $\upsilon \in \mc{C}$. We thus consider the following setting:\\
	
	\noindent
	Let $I \subseteq \R$ be an open interval containing zero. Let $U_0 \subseteq \R^{n-1}$ be an open subset that is diffeomorphic to $\R^{n-1}$. Define $U := I \times U_0$, which is then diffeomorphic to $\R^{n}$. We consider the locally convex Lie algebra $\X_c(U)$ of compactly supported smooth vector fields on $U$. We write $(t, x_1, \ldots, x_{n-1}) \in \R^{n}$ for the coordinates on $\R^{n}$, and $(\partial_t, \partial_{x_1}, \ldots, \partial_{x_{n-1}})$ for the corresponding basis of $\X(\R^{n})$ over $C^\infty(\R^{n})$. Notice that the derivation ${D := [\partial_t, \--]}$ on $\X_c(U)$ does not necessarily integrate to a $1$-parameter group of automorphisms of $\X_c(U)$, because the open set $U$ need not be invariant under the flow of $\partial_t$.

	\begin{lemma}\label{lem: local_case}
		Let $\olpi : \X_c(U) \rtimes \R \partial_t \to \pu(\mcD)$ be a continuous projective unitary representation on the pre-Hilbert space $\mcD$. Assume that
		\begin{equation}\label{eq: commutation_rel_zero_implies_in_kernel}
			[v, D v] = 0 \implies \olpi(Dv) = 0, \qquad \forall v \in \X_c(U).
		\end{equation}
		Then $\X_c(U) \subseteq \ker \olpi$.
	\end{lemma}
	\begin{proof}
		Let $p_0 = (t_0,x_0) \in U = I \times U_0$ be arbitrary. Let $f \in C^\infty_c(I)$ and $w \in \X_c(U_0)$ be s.t.\ $f^\prime(t_0) \neq 0$ and $w(x_0) \neq 0$. Define $v \in \X_c(U)$ by $v(t,x) := f(t)w(x)$ for $t \in I$ and $x \in U_0$. Observe that $Dv(t,x) = f^\prime(t)w(x)$. In particular, $Dv(p_0) \neq 0$ and $[v, Dv](t,x) = f(t)f^\prime(t)[w, w](x) = 0$. It follows using \eqref{eq: commutation_rel_zero_implies_in_kernel} that $Dv \in \ker \olpi$. Let $J \subseteq \X_c(U)$ be the ideal generated by $Dv$. Then $J \subseteq \ker \olpi$. As $Dv(p_0) \neq 0$, it follows using \Fref{prop: ideals_locally} that $\X_c(V) \subseteq J$ for some open neighborhood $V \subseteq U$ of $p_0$. So we have $\X_c(V) \subseteq \ker \olpi$. We have thus shown that any $p \in U$ has a neighborhood $V \subseteq U$ for which $\X_c(V) \subseteq \ker \olpi$. Consequently, if $K \subseteq U$ is a compact subset, we can find a finite open cover $\{U_1, \ldots, U_m\}$ of $K$ with $\X_c(U_k) \subseteq \ker \overline{d\rho}$ for all $k \in \{1, \ldots, m\}$. Using a partition of unity argument, it follows that $\X_K(U) \subseteq \ker \overline{d\rho}$ for any compact set $K \subseteq M$, so that $\X_c(U) \subseteq \ker \overline{d\rho}$.
	\end{proof}
	
	\noindent
	We now return to the global setting, and prove \Fref{thm: gpe_vfields}.
	
	\begin{proof}[Proof of \Fref{thm: gpe_vfields}:]
		Let $p \in U$ and let $\upsilon \in \mc{C}$ satisfy $\upsilon(p) \neq 0$. By assumption, $\olpi$ extends to a continuous projective unitary representation of $\X_c(M) \rtimes \R\upsilon$ that is of generalized positive energy at $\upsilon$, again denoted $\olpi$. Since $\upsilon(p) \neq 0$, we can find an open neighborhood $U_p \subseteq M$ of $p$, an open interval $I \subseteq \R$ containing zero, an open subset $U_0 \subseteq \R^{n-1}$ that is diffeomorphic to $\R^{n-1}$, and a diffeomorphism $\phi : I \times U_0 \to U_p$ such that $\phi_\ast([\partial_t, w]) = [\upsilon, \phi_\ast(w)]$ for all $w \in \X_c(I \times U_0)$ \cite[Thm.\ 9.22]{Lee_smooth_mfds}. So $\phi_\ast$ defines an isomorphism 
		$$\phi_\ast : \X_c(I \times U_0) \rtimes \R \partial_t \to \X_c(U_p) \rtimes \R\upsilon.$$
		In view of Theorem~B, we know that $H^2_{\ct}(\X_c(M), \R) = 0$. As $\olpi$ is of generalized positive energy at $\upsilon$, it follows using \Fref{prop: generalized_pe_cauchy_schwarz} that $[w, Dw] = 0$ implies $\olpi(Dw) = 0$ for any $w \in \X_c(M)$. As a consequence, the pull-back of $\olpi$ along the composition 
		$$\X_c(I \times U_0) \rtimes \R \partial_t \xrightarrow{\phi_\ast}  \X_c(U_p) \rtimes \R\upsilon\hookrightarrow \X_c(M) \rtimes \R\upsilon$$
		satisfies the conditions of \Fref{lem: local_case}, from which it then follows that $\X_c(U_p) \subseteq \ker \olpi$. So any $p \in U$ has an open neighborhood $U_p \subseteq M$ satisfying $\X_c(U_p) \subseteq \ker \olpi$. This implies that $\X_c(U) \subseteq \ker \olpi$.
	\end{proof}

	\subsection{The proof of Theorem~A}\label{sec: thmA}
	
	\noindent
	In this section, we derive Theorem~A as a group-level consequence of \Fref{thm: gpe_vfields}. As a special case, we obtain similar results for KMS and bounded representations.
	
	\begin{proof}[Proof of Theorem A\phantom{}:]
		Since the derived representation $\overline{d\rho} : \X_c(M) \rtimes \R \upsilon \to \pu(\mc{H}^\infty)$ is of generalized positive energy at $\upsilon$, it is so at every $\upsilon^\prime$ in the cone $\mc{C}$ generated by the adjoint orbit of $\upsilon$ in $\X_c(M) \rtimes \R \upsilon$. Since $\upsilon$ is non-zero, there exists some open subset $U_0 \subseteq M$ on which $\upsilon$ is non-vanishing. Then $\Ad_f(\upsilon)$ is non-zero on $f(U_0)$ for every $f \in \Diff_c(M)$. Since $M$ is connected, $\Diff_c(M)$ acts transitively on $M$ (because all orbit are open and therefore also closed, cf.\ \cite[p.\ 22]{Milnor_topology}). Hence $\,\bigcup_{\upsilon^\prime \in \mc{C}} \SET{p\in M \st \upsilon^\prime(p) \neq 0} = M$. We obtain using \Fref{thm: gpe_vfields} that $\X_c(M) \subseteq \ker \overline{d\rho}$. \Fref{cor: kernels} now implies that $\Diff_c(M)_0 \subseteq \ker \overline{\rho}$. 
	\end{proof}

	\begin{corollary}\label{cor: KMS_reps_trivial}
		Suppose that $M$ is connected and that $\dim(M) > 1$. Let $\upsilon \in \X(M)\setminus \{0\}$ be a complete vector field on $M$. Let $\overline{\rho} : \Diff_c(M) \rtimes_\upsilon \R \to \PU(\mc{H})$ be a smooth projective unitary representation that is smoothly-KMS at $\upsilon$ relative to $\Diff_c(M)$. Assume that the von Neumann algebra $\rho(\Diff_c(M))^{\prime \prime}$ is a factor. Then $\Diff_c(M)_0 \subseteq \ker \overline{\rho}$.
	\end{corollary}
	
	\begin{proof}
		Let $\rho : G \to \U(\mc{H})$ be the lift of $\overline{\rho}$, where the Lie group $G$ is a central $\T$-extension of $\Diff_c(M) \rtimes_\upsilon \R$. Let $H \subseteq G$ be the Lie subgroup covering $\Diff_c(M)$. Let $\h$ and $\g$ denote the Lie algebras of $H$ and $G$, respectively. Let $\mc{N} := \rho(H)^{\prime \prime}$ be the von Neumann algebra generated by $\rho(H)$. As $\overline{\rho}$ is smoothly-KMS at $\upsilon$ relative to $\Diff_c(M)$, there is some $\xi \in \g$ covering $\upsilon$ such that $\rho$ is smoothly-KMS at $\xi \in \g$ relative to $H$. Let $\phi$ be a normal state on $\mc{N}$ for which the function $H \to \C, h \mapsto \phi(\rho(h))$ is smooth, and that is $\sigma$-KMS for $\sigma_t(x) = \rho(e^{t\xi})x\rho(e^{-t\xi})$ with $t \in \R$ and $x\in \mc{N}$. Let $\rho_\phi : H \rtimes \R \to \U(\mc{H}_\phi)$ be the associated unitary representation of $H \rtimes \R$ on the GNS-Hilbert space $\mc{H}_\phi$. According to \Fref{thm: kms_of_qpe}, the representation $\rho_\phi$ on $\mc{H}_\phi$ is smooth and of generalized positive energy at $(0,1) \in \h \rtimes \R$. It follows from Theorem~A that $\rho_\phi(H_0) \subseteq \T \id_{\mc{H}_\phi} $, where $H_0$ denotes the identity component of $H$. Because the von Neumann algebra $\mc{N}$ is a factor, the GNS-representation $\mc{N} \to \mathcal{B}(\mc{H}_\phi)$ is injective (see e.g.\ \cite[Rem.\ 5.3 items 1 and 3]{Milan_reps_jets}). It follows that $\rho(H_0) \subseteq \T \id_{\mc{H}}$. Since $H_0$ covers $\Diff_c(M)_0$, this implies that $\Diff_c(M)_0 \subseteq \ker \overline{\rho}$.
	\end{proof}

	\begin{corollary}\label{cor: bdd_reps_trivial}
		Suppose that $\dim(M) > 1$. Let $\overline{\rho} : \Diff_c(M) \to \PU(\mc{H})$ be a smooth projective unitary representation that is bounded, i.e., continuous w.r.t.\ the norm topology on $\PU(\mc{H})$. Then $\Diff_c(M)_0 \subseteq \ker \overline{\rho}$.
	\end{corollary}
	
	\begin{proof}
		Let $\rho : G \to \U(\mc{H})$ be the lift of $\overline{\rho}$, where $G$ is a central $\T$-extension of $\Diff_c(M)$ with Lie algebra $\g$. Let $p \in M$. Take $\upsilon \in \X_c(M)$ with $\upsilon(p) \neq 0$ and let $\xi \in \g$ cover~$\upsilon$. Since $\rho$ is continuous w.r.t.\ the norm-topology on $\U(\mc{H})$, the self-adjoint operator $-i \restr{\frac{d}{dt}}{t=0} \rho(\exp_G(t\xi))$ is bounded. It follows that $\overline{\rho}$ is of (generalized) positive energy at $\upsilon \in \X_c(M)$. Using Theorem~A, this implies that $\upsilon^\prime \in \ker \overline{d\rho}$ for any $\upsilon^\prime \in \X_c(M)$ for which $\supp(\upsilon^\prime)$ is contained in the connected component of $p$ in $M$. As $p \in M$ was arbitrary, we find that $\X_c(M) \subseteq \ker \overline{d\rho}$. We conclude using \Fref{cor: kernels} that $\Diff_c(M)_0 \subseteq \ker \overline{\rho}$.
	\end{proof}
	
	\section{Continuous second Lie algebra cohomology}\label{sec: second_cts_LA_cohom}
	\noindent
	In this section we prove Theorem~B, that is, we determine the continuous second Lie algebra cohomology $H_{\ct}^2(\X_c(M), \R)$.
	This is a crucial ingredient for the results of \Fref{sec: gen_pe_reps_vfields}.\\
	
	\noindent
	In \Fref{sec: cohom_dim_geq2} we consider the proof of Theorem~B for the case ${\dim(M) > 1}$, 
	and in \Fref{sec: cohom_dim_one} we consider the case $\dim(M) = 1$. 
	The cohomology $H^\bullet_\ct(\X_c(M), \R)$ 
	is generally different from the Gelfand--Fuks 
	cohomology $H^\bullet_\ct(\X(M), \R)$ (cf.\ also \cite{Shnider_cohom_cpt_vfields}). In \Fref{sec: relation_with_gf} we therefore clarify 
	the relationship between $H^2_\ct(\X_c(M), \R)$ and the second Gelfand--Fuks cohomology $H^2_\ct(\X(M), \R)$.
	They both vanish when $\mathrm{dim}(M) >1$, 
	but differ for noncompact manifolds of dimension 1.
	\\
	
	\noindent
	We first make some general observations that will be useful for the proof of Theorem~B.
	 
	\begin{definition}
		A $2$-cochain $\psi \colon \X_c(M) \times \X_c(M) \rightarrow \R$ is called \emph{diagonal} if 
		$${\supp(v)\cap \supp(w) = \emptyset} \quad\implies\quad \psi(v,w) = 0, \qquad \forall v,w\in \X_{c}(M).$$
	\end{definition}
	
	\begin{lemma}\label{lem: diagonal}
		Every $2$-cocycle on $\X_c(M)$ is diagonal.
	\end{lemma}
	\begin{proof}
		Let $v,w \in \X_c(M)$ have disjoint support. Then we can find open subsets $U_1, U_2 \subseteq M$ with $U_1 \cap U_2 = \emptyset$ such that $\supp(v) \subseteq U_1$ and $\supp(w) \subseteq U_2$. Since the Lie algebra $\X_c(U_1)$ is perfect (\cite[Thm.\ 1.4.3]{Banyaga_book_diffeo} or \cite[Cor.\ 1]{Bas_inf_nat_bundles}), there exist $v^{i}_1, v^{i}_2 \in \X_c(U_1)$ s.t.\ 
		$v = \sum_{i=1}^{N} [v^{i}_1, v^{i}_2]$. Since $\psi \colon \X_c(M) \times \X_c(M) \rightarrow \R$ satisfies the cocycle identity
		\begin{equation}\label{eq: 2cocycle}
			\psi([u,v], w) = \psi([u,w],v) + \psi(u, [v,w]),
		\end{equation}
		and since $w$ has support disjoint from that of $v^{i}_1$ and $v^{i}_2$, we have
		$$ \psi(v, w) = \sum_{i=1}^{N} \psi([v^{i}_1, v^{i}_2], w) = \sum_{i=1}^{N}\psi([v^{i}_1, w], v^{i}_2) + \psi(v^{i}_1, [v^{i}_2, w]) = 0,$$
	as required.
	\end{proof}
	\noindent
	Let $\g$ be a locally convex Lie algebra, and let $\g'$ be its continuous dual. We consider $\g'$ as a 
	$\g$-module with the coadjoint action, defined by
	$(\xi \cdot \alpha)(\eta) := -\alpha ([\xi, \eta])$ for $\alpha \in \g'$. For a continuous cochain $\psi \in C_{\ct}^q(\g, \R)$ with values in $\R$, 
	we define a (not necessarily continuous) cochain $\hat{\psi}\in C^{q-1}(\g, \g')$ by
	\begin{equation}\label{eq: hat_map}
		\hat{\psi}(\xi_1, \ldots, \xi_{q-1})(\eta) := \psi(\xi_1, \ldots, \xi_{q-1}, \eta).
	\end{equation}
	\noindent
	Let $C^{\bullet -1}(\g, \g')$ denote the shifted complex with $C^{q-1}(\g,\g')$ in degree $q$, 
	and with differential from the $q^{\mrm{th}}$ to the $(q+1)^{\mrm{th}}$ degree given by	$d_{\g} \colon C^{q-1}(\g,\g') \rightarrow C^{q}(\g,\g')$ as in \eqref{eq: differential_la_coh}. We note the following:
	\begin{proposition}
	The assignment $\psi \mapsto \hat{\psi}$ defines a morphism $C^{\bullet}_{\ct}(\g, \R) \rightarrow C^{\bullet -1}(\g, \g')$ of cochain complexes. Moreover, the induced map $H^2_{\ct}(\g, \R) \rightarrow H^1(\g,\g')$ is injective.
	\end{proposition}
	\begin{proof}
	That $\psi \mapsto \hat{\psi}$ defines a morphism $C^{\bullet}_{\ct}(\g, \R) \rightarrow C^{\bullet -1}(\g, \g')$ of cochain complexes follows from a straightforward computation using \eqref{eq: differential_la_coh}. For the final assertion, note that $\psi \mapsto \hat{\psi}$ is injective and restricts to an isomorphism $C^1_{\ct}(\g, \R) \rightarrow C^0(\g,\g')$. This implies that a 2-cocycle $\psi \in C^2_{\ct}(\g,\R)$ is a coboundary if and only if $\hat{\psi} \in C^1(\g,\g')$
	is a coboundary
	\end{proof}
	
	\noindent
	So if $\psi \in C_{\ct}^2(\X_c(M),\R)$ is a continuous 2-cocycle on $\X_c(M)$ with trivial coefficients, 
	then $\hat{\psi} \in C^1(\X_c(M),\X_c(M)')$ is a 1-cocycle with values in the continuous dual space $\X_c(M)'$, and
	\begin{equation}
		\hat{\psi}([v,w]) = v \cdot \hat{\psi}(w) - w \cdot \hat{\psi}(v)
	\end{equation}
	for all $v, w \in \X_{c}(M)$.
	
	\begin{lemma}\label{lem: extension_of_1cocycle}
		Let $\psi \in C_{\ct}^2(\X_c(M))$ be a $2$-cocycle. Then $\hat{\psi} \in C^1(\X_c(M), \X_c(M)^\prime)$ extends to a 
		$1$-cocycle $\hat{\psi} \in C^1(\X(M), \X_c(M)^\prime)$.
	\end{lemma}
	\begin{proof}
		Let $\{K_{i}\}_{i \in \N}$ be an exhaustion of $M$ by compact subsets. For $v \in \X(M)$ and $i \in \N$, we define $\hat{\psi}_i(v) \in \X_{K_i}(M)'$ by 
		$\hat{\psi}_i(v)(w) := \psi(f_{K_i}v, w)$ for an arbitrary $f_{K_i} \in C^\infty_c(M)$ that satisfies $f_{K_i}(x) = 1$ for all $x$ in some open neighborhood $U_i$ of $K_i$. This is independent of $f_{K_i}$ by \Fref{lem: diagonal}. The various $\hat{\psi}_i(v)$ define an element $\hat{\psi}(v)$ of $\X_{c}(M)^\prime$. Indeed,
		$\X_c(M)$ is the locally convex inductive limit $\X_c(M) = \varinjlim_i X_{K_i}(M)$, and
		the functionals $\hat{\psi}_i(v)$ are compatible in the sense that the restriction of $\psi_j(v)$ to $\X_{K_i}(M) \subseteq \X_{K_j}(M)$ coincides with $\psi_i(v)$ if $K_i \subseteq K_j$. The linear map $\hat{\psi} : \X(M) \to \X_c(M)^\prime$ obtained in this way clearly extends the original map $\hat{\psi} : \X_c(M) \to \X_c(M)^\prime$ and because $\psi$ is diagonal, it satisfies the cocycle identity
		\begin{equation}\label{eq: OneCocycle}
			\hat{\psi}([v,w]) = v \cdot \hat{\psi}(w) - w \cdot \hat{\psi}(v)
		\end{equation}
		for the action of $\X(M)$ on $\X_{c}(M)'$ by $(v\cdot \phi) (u) = \phi([-v,u])$. 
	\end{proof}
	
	\begin{remark}\label{rem: Peetre}
	Because the continuous $2$-cocycle $\psi$ is diagonal (\Fref{lem: diagonal}),
	the map ${\hat{\psi} : \X(M) \to \X_c(M)^\prime}$ from \Fref{lem: extension_of_1cocycle}
	is \emph{support decreasing}, in the sense that $\supp(\hat{\psi}(v)) \subseteq \supp(v)$ for any $v \in \X(M)$.
	By Peetre's Theorem, we conclude that the linear map $\hat{\psi} : \X(M) \to \X_c(M)^\prime$ is a 
	differential operator of locally finite degree. This follows from \cite[Thm.~1]{Peetre}; the locally finite set of discontinuous points for 
	$\hat{\psi} \colon \X(M) \rightarrow \X_{c}(M)'$ is empty because $\psi \colon \X_c(M) \times \X_c(M) \rightarrow \R$ 
	is continuous.
	\end{remark}
	
	\subsection{Manifolds $M$ of dimension $\dim(M) > 1$}\label{sec: cohom_dim_geq2}
	\noindent
	We now proceed with the proof of Theorem~B for manifolds of dimension $\dim(M) > 1$. 
	Note that we will occasionally use Einstein summation convention, so repeated indices imply a sum.

	\subsubsection{The local setting}
	\noindent
	We begin with the case where $M = \R^n$. 
	In \Fref{lem: second_la_cohom_vanishes_locally} below, we will show that $H_{\ct}^2(\X_c(\R^n), \R) = \{0\}$ for $n>1$.
	This should be regarded as the local analog of Theorem~B for $\dim(M) > 1$. The analogous statement in Gelfand--Fuks cohomology follows e.g.\ from \cite[Thm.\ 3.12]{Lukas_GF_cohom}.\\
	
	\noindent The following is an adaptation of a result in \cite{BasCorneliaLeo_CEdiv}, and we thank Cornelia Vizman and Leonid Ryvkin for illuminating discussions on this topic.\\
	
	\noindent
	We first make some preliminary observations. 
	The Lie algebra $W_n \subseteq \X(\R^n)$ of vector fields with polynomial coefficients is $\Z$-graded, with $W_n^k$ consisting of vector fields with homogeneous polynomial coefficients of degree $k+1$ for $k\geq -1$, and $W_n^k = \{0\}$ for $k < -1$. Since $[W_n^k, W_n^l] \subseteq W_n^{k+l}$, the constant vector fields $W_n^{-1}$ decrease the degree by $1$. Also, every $W_n^{k}$ is a representation of the Lie algebra $W_n^0$ of linear vector fields, which we identify with $\gl(n, \R)$ via the isomorphism $\gl(n, \R) \to W_n^0$ that maps $(A^{\mu}_{\nu})_{\mu, \nu = 1}^n$ to the linear vector field $a^{\mu}_{\nu}x^{\nu}\partial_{\mu}$ with constant coefficients $a^\mu_\nu = -A^{\mu}_{\nu}$. Under this identification, we have $W_n^k \cong S^{k+1}(\R^d)^\ast \otimes \R^d$ as $\gl(n, \R)$-representation for every $k \in \N_{\geq 0}$, where $S^{k+1}(\R^d)^\ast$ denotes the space of homogeneous polynomials on $\R^n$ of degree $k+1$. The Euler vector field $E = x^{\mu}\partial_{\mu}$ acts on $v\in W_n^{k}$ by $[E,v] = k v$.
	
	\noindent
	\begin{lemma}\label{lem: action_euler_on_translation_inv}
		The space $\X_c(\R^n)^\prime$ of translation invariant elements is equivalent to the space ${(\R^n)^\ast\otimes \wedge^n(\R^n)^\ast}$ as a $\mathfrak{gl}(n,\R)$-representation. In particular, the Euler vector field $E \in W_n^0$ acts on a translation-invariant $\phi \in \X_c(\R^n)^\prime$ by 
		\begin{equation}\label{eq: euler_on_trlns_inv}
			E\cdot \phi = (n+1)\phi.
		\end{equation}
	\end{lemma}
	\begin{proof}
		The linear vector field $a^{\mu}_{\nu}x^{\nu}\partial_{\mu}$ corresponding to $A = (A^{\mu}_{\nu})_{\mu,\nu=1}^n$ acts on $\phi \in \X_c(\R^n)^\prime$ according to 
		\begin{align}\label{eq: action_lin_vfields_on_distr}
			\begin{split}
				(a^{\mu}_{\nu}x^{\nu}\partial_{\mu} \cdot \phi )(u^{\sigma}\partial_{\sigma}) 
				&= \phi(-a^{\mu}_{\nu}[x^{\nu}\partial_{\mu}, u^{\sigma}\partial_{\sigma}]) \\
				&= -\phi(a^\mu_\nu x^\nu (\partial_\mu u^\sigma) \partial_\sigma) + \phi(a^\mu_\sigma u^\sigma \partial_\mu)\\
				&= -\mathrm{tr}(A)\phi(u^{\sigma}\partial_{\sigma}) + \phi(a^{\mu}_{\sigma}u^{\sigma}\partial_{\mu} ) + (\partial_{\mu}\cdot \phi)(a^{\mu}_{\nu}x^{\nu}u^{\sigma}\partial_{\sigma}),
			\end{split}
		\end{align}
		where the last equality uses $(\partial_\mu \cdot \phi)(a^\mu_\nu x^\nu u^\sigma \partial_\sigma) = - \phi(a_\nu^\mu u^\sigma \partial_\sigma) \delta_{\mu}^\nu - \phi(a_\nu^\mu x^\nu (\partial_\mu u^\sigma)\partial_\sigma)$. Assume now that $\phi$ is translation-invariant. Then $(\partial_{\mu}\cdot \phi)(a^{\mu}_{\nu}x^{\nu}u^{\sigma}\partial_{\sigma}) = 0$, and $\phi(u^\sigma \partial_\sigma) = b_\sigma I(u^\sigma)$ for some vector $b = (b_\sigma)_{\sigma=1}^n \in \R^n$, where $I(f) := \int_{\R^n} fdx$ for $f \in C^\infty_c(\R^n)$. It follows using \eqref{eq: action_lin_vfields_on_distr} that
		$$ (a^{\mu}_{\nu}x^{\nu}\partial_{\mu} \cdot \phi )(u^{\sigma}\partial_{\sigma}) = \big(-\mrm{tr}(A) b_\sigma + a_\sigma^\mu b_\mu\big) I(u^\sigma) = b^\prime_\sigma I(u^\sigma), $$
		where $b^\prime := -\mrm{tr}(A^T) b - A^T b$. This corresponds to the natural action of $\gl(n, \R)$ on $(\R^n)^\ast\otimes \wedge^n(\R^n)^\ast$ under the isomorphism $\gl(n, \R) \cong W_n^0$ specified above, so the assertion follows.
	\end{proof}
	
	\begin{lemma}\label{lem: second_la_cohom_vanishes_locally}
		Let $n > 1$ be an integer. Then $H_{\ct}^2(\X_c(\R^n), \R) = \{0\}$.
	\end{lemma}
	\begin{proof}
		Let $\psi$ be a continuous $2$-cocycle on $\X_{c}(\R^n)$, and let $\hat{\psi}$ 
		be the corresponding 1-cocycle in $C^1(\X(\R^n), \X_c(\R^n)^\prime)$, obtained using \Fref{lem: extension_of_1cocycle}. By \Fref{rem: Peetre}, we can expand $\hat{\psi}$ into a locally finite sum as
		\begin{equation}\label{eq: DOexpansion}
			\hat{\psi}(v) = \sum_{\vec{\sigma} \in \mathbb{N}_{\geq 0}^{n}}
			\left(\frac{\partial^{|\vec{\sigma}|}}{\partial x^{\vec{\sigma}}} v^{\mu}\right) 
			\phi^{\vec{\sigma}}_{\mu},
		\end{equation}
		where $\phi^{\vec{\sigma}}_{\mu} \in \X_{c}(\R^n)^\prime$. Here 
		$\frac{\partial^{|\vec{\sigma}|}}{\partial x^{\vec{\sigma}}} := (\frac{\partial}{\partial x^{1}})^{\sigma_1} \cdots (\frac{\partial}{\partial x^n})^{\sigma_n}$ is a higher order partial derivative.
		We show for any integer $k \geq -1$ that $\hat{\psi}$ is cohomologous to a 1-cocycle that vanishes on the subspace $W_n^{\leq k}$ of vector fields with polynomial coefficients 
		of degree at most $k+1$.
		
		\paragraph{The case $k=-1$.}
		The cocycle identity \eqref{eq: OneCocycle} for constant vector fields $v = \partial_{\mu}$ and $w = \partial_{\nu}$ yields 
		\begin{equation}\label{eq:distributionclosed}
			\partial_{\mu} \cdot \phi^{\vec{0}}_{\nu} - \partial_{\nu}\cdot\phi^{\vec{0}}_{\mu} = 0.
		\end{equation}
		We identify $\X_{c}(\R^n)' \simeq \mc{D}^\prime(\R^n)\otimes (\R^n)^\ast$ with $n$ copies of the
		distributions $\mc{D}^\prime(\R^n)$ by setting 
		$$(\zeta_{\sigma}\otimes dx^{\sigma}) (v^{\mu}\partial_{\mu}) := \zeta_{\sigma}(v^{\sigma}), \qquad \text{ for }\zeta_\sigma \in \mcD^\prime(\R^n) \text{ and } v^\mu \in C^\infty_c(\R^n).$$
		The action of $\partial_{\mu}$ on $\X_{c}(\R^n)'$ is then simply given by differentiating the components in $\mcD^\prime(\R^n)$, so that for $\phi^{\vec{0}}_{\nu} =  \phi^{\vec{0}}_{\nu\sigma} \otimes dx^{\sigma}$ we have $\partial_{\mu}\cdot \phi^{\vec{0}}_{\nu} = \partial_{\mu}\phi^{\vec{0}}_{\nu\sigma} \otimes dx^{\sigma}$. Indeed, we compute that
		\begin{align*}
			(\partial_{\mu}\cdot\phi^{\vec{0}}_{\nu})(X^{\tau}\partial_{\tau}) 
			&= \phi^{\vec{0}}_{\nu}(-[\partial_{\mu}, X^{\tau}\partial_{\tau}]) = \phi^{\vec{0}}_{\nu\sigma}\otimes dx^{\sigma}(-(\partial_{\mu}X^{\tau})\partial_{\tau}) \\
			&= \phi^{\vec{0}}_{\nu\sigma}(-(\partial_{\mu}X^{\sigma})) = (\partial_{\mu}\phi^{\vec{0}}_{\nu\sigma})(X^{\sigma}) =  (\partial_{\mu}\phi^{\vec{0}}_{\nu\sigma} \otimes dx^{\sigma})(X^{\tau}\partial_{\tau}).	
		\end{align*}
		\Fref{eq:distributionclosed} therefore yields for each $\sigma$ that $\partial_{\mu}\cdot\phi^{\vec{0}}_{\nu\sigma} - \partial_{\nu}\cdot\phi^{\vec{0}}_{\mu\sigma} = 0$ for all integers $1 \leq \mu, \nu \leq n$. With respect to the differential $d : \Omega_c^{n-p}(\R^n)^\prime \to \Omega_c^{n - (p+1)}(\R^n)^\prime,\; \langle d T, \alpha\rangle := (-1)^{p+1} \langle T, d\alpha\rangle$ on the space of currents (cf.\ \cite[III\S11]{deRham_book_1984}), this means precisely that the current $c_\sigma := \phi_{\mu \sigma}^{\vec{0}} dx^\mu \in \Omega_c^{n-1}(\R^n)^\prime$ is closed, because 
		$$d c_\sigma = \sum_{1 \leq \mu < \nu \leq n} (\partial_\mu \cdot \phi_{\nu \sigma}^{\vec{0}} - \partial_\nu \cdot \phi_{\mu \sigma}^{\vec{0}}) dx^\mu \wedge dx^\nu = 0.$$
		(Identifying $\mcD^\prime(\R^n) \cong \Omega_c^n(\R^n)^\prime$ using the volume form $dx^1 \wedge \cdots \wedge dx^n$ on $\R^n$, for any $T \in \mcD^\prime(\R^n)$ and $\alpha \in \Omega^p(\R^n)$, we interpret $T\alpha$ as element of $\Omega_c^{n-p}(\R^n)^\prime$ via the pairing $\langle T\alpha, \beta\rangle := T(\alpha \wedge \beta)$ for $\beta \in \Omega_c^{n-p}(\R^n)$, cf.\ \cite[p.\ 36]{deRham_book_1984}.) By the Poincar\'e Lemma for currents \cite[IV\S19]{deRham_book_1984}, it follows that there exist distributions $\eta_\sigma \in \mc{D}^\prime(\R^n)$ with $\partial_{\mu} \eta_{\sigma} = \phi^{\vec{0}}_{\mu\sigma}$ for all integers $1 \leq \mu, \sigma \leq n$. The $1$-coboundary $d_\g(\eta_\sigma \otimes dx^\sigma)$ in $C^1(\X(\R^n), \X_c(\R^n)^\prime)$ thus agrees with $\hat{\psi}$ on $\partial_\mu \in W_n^{-1}$:
		$$ d_\g(\eta_\sigma \otimes dx^\sigma)(\partial_\mu) = \partial_\mu \cdot (\eta_\sigma \otimes dx^\sigma) = \partial_\mu \eta_\sigma \otimes dx^\sigma = \phi^{\vec{0}}_{\mu\sigma} \otimes dx^\sigma = \phi^{\vec{0}}_{\mu} = \hat{\psi}(\partial_\mu) $$
		Replacing $\hat{\psi}$ by the $1$-cocycle $\hat{\psi} - d_\g(\eta_{\sigma}\otimes dx^{\sigma})$, we assume from now on that $\hat{\psi}$ vanishes on $W_n^{-1}$.
		
		\paragraph{The case $0 \leq k \leq n$.}
		Suppose that $\hat{\psi}$ vanishes on $W_{n}^{\leq(k-1)}$. Let $v \in W_{n}^{k}$. Since $[\partial_{\mu}, v] \in W_{n}^{k-1}$, the cocycle identity \eqref{eq: OneCocycle} yields $\partial_{\mu} \cdot \hat{\psi}(v) = 0$ for all $\mu$, so that $\hat{\psi}(v)$ is translation invariant. So $E\cdot \hat{\psi}(v) = (n+1)\hat{\psi}(v)$, in view of \Fref{lem: action_euler_on_translation_inv}. From the cocycle identity $\hat{\psi}([E,v]) = E\cdot \hat{\psi}(v) - v\cdot \hat{\psi}(E)$, we find for any $v \in W_n^{k}$ that
		\begin{equation}\label{eq:EulerFixesEverything}
			(n + 1 - k)\hat{\psi}(v) = v \cdot \hat{\psi}(E). 
		\end{equation}
		We consider separately the cases $k = 0$ and $0 < k \leq n$. Suppose that $k = 0$. The preceding then shows that $\hat{\psi}(v) = \frac{1}{n+1} v \cdot \hat{\psi}(E)$ for all $v \in W_n^{0}$. The $0$-cochain $\eta = \frac{1}{n+1}\hat{\psi}(E)$ therefore satisfies $(d_\g \eta)(v) = \hat{\psi}(v)$ for any $v \in W_n^{0}$. Since $E \in W_n^0$, we know that $\hat{\psi}(E)$ is translation-invariant, so we also have $(d_\g \eta)(v) = \frac{1}{n+1}v\cdot \hat{\psi}(E) = 0$ for $v \in W_{n}^{-1}$. Replacing $\hat{\psi}$ by the cohomologous cocycle $\hat{\psi} - d_\g\eta$ if necessary, we may assume that $\hat{\psi}$ vanishes on $W_{n}^{\leq 0}$. Suppose next that $0 < k \leq n$. Then $E \in W_n^{\leq (k-1)}$, so $\hat{\psi}(E) = 0$. Consequently, \eqref{eq:EulerFixesEverything} implies that $\hat{\psi}(v) = 0$ for any $v \in W_n^k$ and hence $\hat{\psi}$ vanishes on $W_{n}^{\leq k}$. Inductively, we thus find that $\hat{\psi}$ vanishes on $W_{n}^{\leq n}$, and that $\hat{\psi}(v)$ is translation invariant for any $v\in W_{n}^{n+1}$.
		
		\paragraph{The case $k = n+1$.} The cocycle identity \eqref{eq: OneCocycle} for $A\in W_n^{0}$ and $v\in W_{n}^{n+1}$ reads $\hat{\psi}([A,v]) = A \cdot \hat{\psi}(v)$, because $\hat{\psi}(A) = 0$. Since $\hat{\psi}(v)$ is translation invariant for any $v \in W_{n}^{n+1}$, we conclude using \Fref{lem: action_euler_on_translation_inv} that the linear map 
		$$\restr{\hat{\psi}}{W_{n}^{n+1}} : W_{n}^{n+1} \rightarrow (\R^n)^\ast\otimes \wedge^{n}(\R^n)^\ast \subseteq \X_c(\R^n)^\prime$$
		is an intertwiner of $\mathfrak{gl}(n,\R)$-representations. The action of $\mf{sl}(n, \R)$ on $ \wedge^{n}(\R^n)^\ast$ is trivial, and we notice that
		$$ \Hom_{\mf{sl}(n, \R)}\left(W_{n}^{n+1}, (\R^n)^\ast\right) \cong \Hom_{\mf{sl}(n,\R)}\left(S^{n+2}({\R^n})^\ast, (\R^n)^\ast \otimes (\R^n)^\ast  \right) = 0,$$
		because $(\R^n)^\ast \otimes (\R^n)^\ast \cong S^2(\R^n)^\ast \oplus \bigwedge^2(\R^n)^\ast$ does not contain the irreducible $\mf{sl}(n,\R)$-representation on $S^{n+2}({\R^n})^\ast$ (cf.\ \cite[Prop.\ 15.15]{Fulton_Harris}). So $\hat{\psi}(v) = 0$ for any $v \in W_{n}^{n+1}$.
		
		\paragraph{The case $k > n+1$.}
		Suppose that $\hat{\psi}$ vanishes on $W_{n}^{\leq(k-1)}$ for $k > n + 1$. Then \eqref{eq:EulerFixesEverything} implies that $\hat{\psi}(v) = 0$ for any $v \in W_n^k$, so $\hat{\psi}$ vanishes on $W_{n}^{\leq k}$. Inductively, we thus find that $\hat{\psi}$ vanishes on $W_{n}^{\leq k}$ for any integer $k \geq -1$. This implies that all the coefficients $\phi^{\vec{\sigma}}_{\mu}$ in \fref{eq: DOexpansion} are zero, so $\hat{\psi}= 0$ and hence $\psi = 0$.
	\end{proof}

	\subsubsection{A local-to-global argument}
	
	\noindent
	Having established that $H^2_\ct(\X_c(\R^n), \R) = \{0\}$, we now show that $H^2_\ct(\X_c(M), \R)$ 
	vanishes for general manifolds $M$ of dimension greater than $1$, using a local-to-global argument. This completes the proof of the first part of Theorem~B.\\
	
	\noindent 
	We let $\X_c^\prime$ denote the presheaf defined by $U \mapsto \X_c(U)^\prime$ and the natural restriction maps. This is in fact an acyclic sheaf by \Fref{prop: distr_acyclic_sheaf}.
	
	\begin{lemma}\label{lem: second_la_cohom_vanishes}
		Assume that $\dim(M) > 1$. Then $H_{\ct}^2(\X_c(M), \R) = \{0\}$.
	\end{lemma}
	\begin{proof}
		Let $n := \dim(M)$. The continuous Chevalley--Eilenberg cochains define a presheaf $U \mapsto C_{\ct}^m(\X_{c}(U), \R)$ for any $m \in \N$, that we denote by $C_{\ct}^m(\X_{c})$. We denote by $Z_{\ct}^m(\X_{c}) \subseteq C_{\ct}^m(\X_{c})$ its sub-presheaf consisting of cocycles. Let $\mathcal{U} = \{U_i \st i\in S\}$ be an open cover of $M$ such that every $U_i$ is diffeomorphic to $\R^{n}$, and consider the (augmented) double complex $\check{C}^{\bullet}(\mc{U}, C_{\ct}^{\bullet}(\X_{c}))$ for the \v{C}ech-cohomology with coefficients in the presheaf $C_{\ct}^{\bullet}(\X_{c})$. Restricted to cocycles in Chevalley--Eilenberg degree $2$, the left lower portion looks as follows:
		
		\begin{center}
			\begin{tikzcd}
				& 0 & 0 & 0\\
				0\arrow[r]& Z_{\ct}^2(\X_{c}(M))\arrow[u]\arrow[r,"\check{\delta}"] &\prod_{i\in S} Z_{\ct}^2(\X_{c}(U_i)) \arrow[u]\arrow[r, "\check{\delta}"] 
				& \prod_{i, j\in S} Z_{\ct}^2(\X_{c}(U_i\cap U_j))\arrow[u]\\
				0\arrow[r]& C_{\ct}^1(\X_{c}(M)) \arrow[u,"d_\g"]\arrow[r,"\check{\delta}"]& \prod_{i\in S} C_{\ct}^1(\X_{c}(U_i)) \arrow[u,"d_\g"]\arrow[r, "\check{\delta}"]
				& \prod_{i, j\in S} C_{\ct}^1(\X_{c}(U_i\cap U_j))
				\arrow[u,"d_\g"]\\
				& 0 \arrow[u] & 0\arrow[u] & 0.\arrow[u]
			\end{tikzcd}
		\end{center}
		The middle column is exact by \Fref{lem: second_la_cohom_vanishes_locally}, as every $U_i \in \mc{U}$ is diffeomorphic to $\R^n$, and the column on the right is exact at $\prod_{i, j\in S} C_{\ct}^1(\X_{c}(U_i\cap U_j))$ because $\X_c(U_i \cap U_j)$ is perfect for any $i,j\in S$ \cite[Thm.\ 1.4.3]{Banyaga_book_diffeo}. The bottom row is exact because $C_{\ct}^1(\X_c) = \X_{c}^\prime$ is an acyclic sheaf, by \Fref{prop: distr_acyclic_sheaf}. \Fref{lem: diagonal} further guarantees that the map $\check{\delta} \colon Z_{\ct}^2(\X_{c}(M)) \rightarrow \prod_{i\in S}Z_{\ct}^2(\X_{c}(U_i))$ is injective. Indeed, suppose that $\psi(\X_c(U_i), \X_c(U_i)) = \{0\}$ for all $i \in S$. Then $\psi(X_c(U_i), \X_c(M)) = \{0\}$ for any $\psi \in C_{\ct}^2(\X_c(M))$ and $i \in S$, because $\psi$ is diagonal. So $\psi = 0$, by a partition of unity argument. A straightforward diagram chase now shows that $H_\ct^2(\X_{c}(M),\R)$ vanishes.\qedhere \\
	\end{proof}
	
	\subsection{Manifolds $M$ of dimension one}\label{sec: cohom_dim_one}
	
	\noindent
	Having proven Theorem~B for manifolds of dimension greater than $1$, we proceed with the remaining case, and determine
	$H^2_\ct(\X_c(M), \R)$ for manifolds of dimension 1. In the connected case, $M$ must be diffeomorphic to either $\R$ or $S^1$. It is well-known that $H_{\ct}^2(\X(S^1), \R) = \R$ is spanned by the class of the Virasoro cocycle (cf.\ \cite[Prop.\ 2.3]{Khesin_book}):
	$$ \psi_{\mathrm{vir}}(f\partial_\theta, g\partial_\theta) = \int_{S^1} f^{\prime \prime \prime}(\theta) g(\theta)d\theta, \qquad f,g \in C^\infty(S^1).$$
	A slight adaptation of the proof of \Fref{lem: second_la_cohom_vanishes_locally} allows us to prove the analogous result on the real line:
	\begin{lemma}\label{lem: second_cohom_on_line}
		The second Lie algebra cohomology $H_{\ct}^2(\X_c(\R), \R)$ is $1$-dimensional. It is spanned by the class of the Virasoro cocycle 
		\begin{equation}\label{eq: vir_cocycle_reals}
			\psi_{\mathrm{vir}}(f\partial_x, g\partial_x) = \int_{\R}f'''(x)g(x)dx, \qquad f,g \in C^\infty_c(\R). 
		\end{equation}
	\end{lemma}
	\begin{proof}
		Let us first observe that the cocycle
		\[
		\psi_{\mathrm{vir}}(f\partial_x, g\partial_x) = \int_{\R}f'''(x)g(x)dx 
		\]
		is not a coboundary. Indeed, if $\eta \in C_{\ct}^1(\X_c(\R), \R)$ is a 1-cochain, then the map $\widehat{d_\g \eta} : \X(\R) \to \X_c(\R)^\prime$ obtained using \Fref{lem: extension_of_1cocycle} is the first-order differential operator $\widehat{d_\g \eta}(f\partial_x) = f (\partial_x \cdot \eta) + 2 f' \eta$. Indeed, this follows from the calculation
		\begin{align*}
			\widehat{d_\g \eta}(f \partial_x)(g\partial_x) 
			&= -\eta([f\partial_x, g\partial_x]) = \eta(f^\prime g \partial_x - fg^\prime \partial_x) \\
			&= 2\eta(f^\prime g \partial_x) - \eta((fg)^\prime \partial_x) = (2f^\prime \eta + f(\partial_x \cdot \eta))(g \partial_x)
		\end{align*}
		for $f \in C^\infty(\R)$ and $g \in C^\infty_c(\R)$. On the other hand, $\hat{\psi}_{\mathrm{vir}}$ is the third-order differential operator $\hat{\psi}_{\mathrm{vir}}(f\partial_x) = f''' I$, where 
		$I \in \X_{c}(\R)'$ is defined by $I(f\partial_x) = \int_{\R}f(x)dx$. So $\psi_{\mathrm{vir}}$ can not be a coboundary. \\
		
		\noindent
		Let $\psi$ be a continuous $2$-cocycle. Let $\hat{\psi} \in C^1(\X(\R), \X_c(\R)^\prime)$ be the corresponding $1$-cocycle obtained using \Fref{lem: extension_of_1cocycle}. We show that $\hat{\psi}$ is cohomologous to a $1$-cocycle in $C^1(\X(\R), \X_c(\R)^\prime)$ that vanishes on the subspace $W_1^{-1} = \R\partial_x$. Choose $\chi \in C^{\infty}_{c}(\R)$ with $\int_{\R}\chi(x)dx = 1$. For any $f \in C^\infty_c(\R)$, the function $$P(f)(x) := \int_{-\infty}^{x} \big(f(s) - I(f \partial_x)\chi(s)\big) ds$$ is smooth and compactly supported. We moreover have $P(f^\prime) =f$, because $I(f^\prime \partial_x) = 0$. Observe that the $0$-cochain $\eta \in \X_c(\R)^\prime = C^0(\X(\R), \X_c(\R)^\prime)$ defined by $\eta(f\partial_x) := \hat{\psi}(\partial_x)(P(f)\partial_x)$ satisfies $\hat{\psi}(\partial_x) +(d_\g\eta)(\partial_x) = 0$, because
		$$(d_\g\eta)(\partial_x)(f \partial_x) = -\eta(f^\prime \partial_x) = -\hat{\psi}(\partial_x)(P(f^\prime)\partial_x) = -\hat{\psi}(\partial_x)(f\partial_x), \qquad \forall f \in C^\infty_c(\R).  $$
		Replacing $\hat{\psi}$ by $\hat{\psi} + d_\g \eta$, we assume from now on that $\hat{\psi}$ vanishes on $W_1^{-1} = \R \partial_x$. Following the case $0 \leq k \leq n$ in the proof of \Fref{lem: second_la_cohom_vanishes_locally}, we may then further assume that $\hat{\psi}$ vanishes on $W_1^{\leq 1}$ and that $\hat{\psi}(x^3 \partial_x) \in \X_c(\R)^\prime$ is translation invariant. The latter implies that $\hat{\psi}(x^3\partial_x) = cI$ for some constant $c \in \R$. It follows that the $1$-cocycle $\hat{\psi} - c\hat{\psi}_{\mathrm{vir}} \in C^1(\X(\R), \X_c(\R)^\prime)$ vanishes on $W_1^{\leq 2}$. Following the case $k > n+1$ in the proof of \Fref{lem: second_la_cohom_vanishes_locally}, this implies that $\hat{\psi} - c\hat{\psi}_{\mathrm{vir}}$ vanishes on $W_1^{\leq k}$ for any integer $k \geq -1$ and therefore that $\hat{\psi} = c\hat{\psi}_{\mathrm{vir}}$. Hence $\psi = c \psi_{\mrm{vir}}$.
	\end{proof}

	\noindent
	We have thus shown that $H^2_{\ct}(\X_c(M), \R) \cong \R$ for any connected $1$-dimensional manifold. Combined with \Fref{lem: second_la_cohom_vanishes}, the following now completes the proof of Theorem~B.
	
	\begin{lemma}\label{lem: one_dim_full}
		Let $M$ be a smooth manifold of dimension $1$. Then 
		\[H^2_{\ct}(\X_c(M), \R) = H^0_{\mathrm{dR}}(M).\]
	\end{lemma}
	\begin{proof}
		Let $\{M_\alpha\}_{\alpha \in \mc{I}}$ be the set of connected components of $M$, where $\mc{I}$ is some countable indexing set. (Here we used that $M$ is second-countable.) As the support of a compactly supported vector field on $M$ intersects only finitely many $M_\alpha$ non-trivially, $\X_c(M)$ is isomorphic to the locally convex direct sum $\X_c(M) \cong \bigoplus_{\alpha \in \mc{I}} \X_c(M_\alpha)$. Hence $\X_c(M)^\prime \cong \prod_{\alpha \in \mc{I}} \X_c(M_\alpha)^\prime$. Furthermore, any $2$-cocycle $\psi : \X_c(M) \times \X_c(M) \to \R$ is diagonal by \Fref{lem: diagonal}, and therefore decomposes as $\psi = \sum_\alpha \psi_\alpha$ for some $2$-cochains $\psi_\alpha$ on $\X_c(M_\alpha)$. Moreover, $\psi$ is a cocycle (or a coboundary) if and only if every $\psi_\alpha$ is so. It follows that 
		$$H^2_{\ct}(\X_c(M), \R) = \prod_{\alpha \in \mc{I}} H^2_{\ct}(\X_c(M_\alpha),\R) = \prod_{\alpha\in \mc{I}} \R \cong H^0_{\mrm{dR}}(M).$$
	\end{proof}
	
	\subsection{Relation with Gelfand--Fuks cohomology}\label{sec: relation_with_gf}
	
	\noindent
	Finally, let us consider the relationship between $H^2_\ct(\X_c(M), \R)$ and the Gelfand--Fuks cohomology $H^2_\ct(\X(M), \R)$. The continuous injection $\X_c(M) \hookrightarrow \X(M)$ induces a natural morphism $C_{\ct}^\bullet(\X(M), \R) \to C_{\ct}^\bullet(\X_c(M), \R)$ of cochain complexes, which descends to a linear map $H^\bullet_\ct(\X(M), \R) \to H_{\ct}^\bullet(\X_c(M), \R)$ on cohomology.
	We show that 
	this map is injective in degree 2. We also show that this map is in general not surjective, so that the continuous cohomology 
	of the compactly supported vector fields is different from Gelfand--Fuks cohomology.\\

	\noindent If $\psi \in C_{\ct}^2(\X_c(M), \R)$ is a diagonal $2$-cochain, its \textit{support} $\supp(\psi)$ is the set of points $x \in M$ with the property that for any neighborhood $U$ of $x$, there exist $v,w \in \X_c(U)$ with $\psi(v,w) \neq 0$. If $x \notin \supp(\psi)$ and $U$ is a neighborhood of $x$ with $\psi(\X_c(U), \X_c(U)) = \{0\}$, then $\psi(\X_c(U), \X_c(M)) = \{0\}$, because $\psi$ is diagonal. The following is a straightforward adaptation of \cite[Lem.\ 4.19]{BasCornelia_ce_ham} to the present setting:
	
	\newpage
	\begin{proposition}\label{prop: relate_cpt_gelfand_fuks} Let $M$ be a smooth manifold.
		\begin{enumerate}
			\item A continuous $2$-cocycle $\psi \in C_{\ct}^2(\X_c(M), \R)$ extends to a continuous $2$-cocycle on $\X(M)$ if and only if it has compact support.
			\item Assume that the $2$-cocycle $\psi \in C_{\ct}^2(\X_c(M), \R)$ has compact support and satisfies $\psi = d_\g \eta$ for some $\eta \in \X_c(M)^\prime$. Then $\supp(\eta) = \supp(\psi)$ and $\eta$ extends to a continuous linear map $\X(M) \to \R$.
			\item The canonical linear map $H^2_\ct(\X(M), \R) \to H^2_\ct(\X_c(M), \R)$ is injective.
		\end{enumerate}		
	\end{proposition}
	\begin{proof}~
		\begin{enumerate}
			\item Assume that $\psi$ has compact support, say $\supp(\psi) = K$. Consider the $1$-cocycle $\hat{\psi} \in C^1(\X(M), \X_c(M)')$ obtained from \Fref{lem: extension_of_1cocycle}. Let $\chi \in C^\infty_c(M)$ satisfy $\restr{\chi}{U} = 1$ for some open neighborhood $U$ of $K$. Define a bilinear map $\widetilde{\psi} : \X(M) \times \X(M) \to \R$ extending $\psi$ by setting $\widetilde{\psi}(v,w) := \hat{\psi}(v)(\chi w) = \psi(\chi v, \chi w)$. This is independent of the choice of $\chi$ because $\psi$ has support $K$. It is moreover continuous, in view of the continuity of both $\psi$ and the map $\X(M) \to \X_c(M), w \mapsto \chi w$. We next show that $\widetilde{\psi}$ is a $2$-cocycle. Observe that $\hat{\psi}(u)(\chi [v,w]) = \hat{\psi}(u)([v, \chi w])$ for any $u,v,w \in \X(M)$, because $\mc{L}_v(\chi)w \in \X_c(M)$ vanishes on a neighborhood of $K$, so that $\hat{\psi}(u)(\mc{L}_v(\chi)w) = 0$. Using \eqref{eq: OneCocycle}, we therefore have
			\begin{align*}
				\widetilde{\psi}([u,v],w) + \widetilde{\psi}(v,[u,w]) 
				&= \hat{\psi}([u,v])(\chi w) + \hat{\psi}(v)(\chi [u,w])\\
				&= \hat{\psi}(u)([v, \chi w ]) - \hat{\psi}(v)([u, \chi w]) + \hat{\psi}(v)(\chi[u,w]) \\
				&= \hat{\psi}(u)(\chi [v,  w ])\\
				&= \widetilde{\psi}(u, [v,w]).			
			\end{align*}
			Conversely, assume that $\psi$ extends to a continuous $2$-cocycle on $\X(M)$, again denoted $\psi$. Suppose that $K := \supp(\psi)$ is not compact. Then we can find a countably infinite sequence $(x_i)_{i \in \N}$ in $K$ of distinct points which has no convergent subsequence. Let $\{U_i \}_{i \in \N}$ be a collection of pairwise disjoint open subsets of $M$ so that $x_i \in U_i$ for all $i \in \N$. Since $x_i \in K$, there exist for every $i \in \N$ some $v_i, w_i \in \X_c(U_i)$ satisfying $\psi(v_i, w_i) = 1$. Notice that $v := \sum_{i=1}^\infty v_i$ and $w := \sum_{i=1}^\infty w_i$ are well-defined smooth vector fields on $M$, because the open sets $U_i$ are pairwise disjoint. Since $\psi \in C_{\ct}^2(\X(M), \R)$ is diagonal and continuous, we obtain the evident contradiction that
			$$ \lim_{N \to \infty} N = \lim_{N \to \infty} \sum_{i=1}^N \psi(v_i, w_i) = \lim_{N \to \infty}  \psi\left(\sum_{i=1}^N v_i, \sum_{i=1}^N w_i\right) =  \psi(v, w) < \infty.$$
			So $\supp(\psi)$ must be compact.
			\item Let $x \notin \supp(\psi)$. Then there exists an open neighborhood $U$ of $x$ such that $\psi(\X_c(U), \X_c(U)) = \{0\}$. Let $u \in \X_c(U)$. Since $\X_c(U)$ is perfect (\cite[Thm.\ 1.4.3]{Banyaga_book_diffeo}), there exist $N \in \N$ and $v_i, w_i \in\X_c(U)$ for $i \in \{1, \ldots, N\}$ s.t.\ $u = \sum_{i=1}^N [v_i, w_i]$. Then $\eta(u) = \sum_{i=1}^N \eta([v_i, w_i]) = - \sum_{i=1}^N \psi(v_i, w_i) = 0$. Thus $\eta$ vanishes on $\X_c(U)$, and $x \notin \supp(\eta)$. It follows that $\supp(\eta) \subseteq \supp(\psi)$. Conversely, suppose that $x \notin \supp(\eta)$. Then there exists an open neighborhood $U$ of $x$ such that $\eta(\X_c(U)) = \{0\}$. Then $\psi = d_\g \eta$ implies that $\psi(\X_c(U), \X_c(U)) =\{0\}$. Thus $x \notin \supp(\psi)$. Hence $\supp(\eta) = \supp(\psi) =: K$. As $\eta$ has compact support $K$, it admits a continuous linear extension $\widetilde{\eta}$ to $\X(M)$ by setting $\widetilde{\eta}(v) := \eta(\chi v)$ for any $\chi \in C^\infty_c(M)$ satisfying $\restr{\chi}{V} = 1$ for some open neighborhood $V$ of $K$. Notice that $\widetilde{\eta}$ is indeed well-defined and continuous.
			\item Let $\psi \in C_{\ct}^2(\X(M), \R)$ be a $2$-cocycle and assume that $\eta \in \X_c(M)^\prime$ satisfies $\psi(v,w) = - \eta([v,w])$ for all $v,w \in \X_c(M)$. The previous items ensure that $\eta$ extends to a continuous functional on $\X(M)$. As $\X_c(M)$ is dense in $\X(M)$ and $\psi$ is continuous on $\X(M) \times \X(M)$, it follows that $\psi(v,w) = - \eta([v,w])$ for all $v,w \in \X(M)$. So $\psi = d_\g \eta$. Hence $[\psi] = 0$ in $H_\ct^2(\X(M), \R)$.\qedhere
		\end{enumerate}		
	\end{proof}
	
	\noindent
	\Fref{prop: relate_cpt_gelfand_fuks}, \Fref{lem: second_la_cohom_vanishes} and \Fref{lem: second_cohom_on_line} have the following consequence for Gelfand--Fuks cohomology:

	\begin{corollary}\label{cor: gf_cohom} Let $M$ be a smooth manifold.
		\begin{enumerate}
			\item If $\dim(M) > 1$,  then $H^2_\ct(\X(M), \R) = 0$. 
			\item If $\dim(M) = 1$, then $H^2_\ct(\X(M), \R) = H^0_{c}(M)$ is the compactly supported de Rham cohomology of $M$ in degree 0. In particular, $H^2_\ct(\X(\R), \R) = 0$. 
		\end{enumerate}
	\end{corollary}
	\begin{proof}~
		\begin{enumerate}
			\item Assume that $\dim(M) > 1$. Then $H^2_\ct(\X_c(M), \R) = 0$ by Theorem~B. We know using \Fref{prop: relate_cpt_gelfand_fuks} that the linear map $H^2_\ct(\X(M), \R) \to H^2_\ct(\X_c(M), \R)$ is injective. It follows that $H^2_\ct(\X(M), \R) = 0$.
			\item By reasoning similar to that in the proof of \Fref{lem: one_dim_full}, it suffices to consider the case where $M$ is connected, so that $M$ is either $S^1$ or $\R$. Since $H^2_\ct(\X(S^1), \R) \cong \R$ \cite[Prop.\ 2.3]{Khesin_book}, it remains to show that $H^2_\ct(\X(\R), \R) = 0$. By \Fref{lem: second_cohom_on_line} we know that $H^2_\ct(\X_c(\R), \R) \cong \R$, which by \Fref{prop: relate_cpt_gelfand_fuks} implies that $H^2_\ct(\X(\R),\R)$ is at most one-dimensional. The non-trivial class in $H^2_\ct(\X_c(\R),\R)$ is spanned by the cocycle $\psi_{\mathrm{vir}}$, defined by \eqref{eq: vir_cocycle_reals}. Assume that $\psi \in C_{\ct}^2(\X(\R),\R)$ is a $2$-cocycle on $\X(\R)$ whose restriction $r(\psi)$ to $\X_c(M) \times \X_c(M)$ is cohomologous to $\psi_{\mathrm{vir}}$. Then $r(\psi) = \psi_{\mathrm{vir}} + d_\g \eta$ for some $\eta \in \X_c(\R)^\prime$. By \Fref{prop: relate_cpt_gelfand_fuks}, we know that $r(\psi)$ has compact support. Consider the associated map $\widehat{r(\psi)} : \X(\R) \to \X_c(\R)^\prime$. We saw in the proof of \Fref{lem: second_cohom_on_line} that $\widehat{d_\g \eta}(f \partial_x) = f (\partial_x \cdot \eta) + 2f^\prime \eta$, and that $\hat{\psi}_{\mathrm{vir}}(f \partial_x) = f''' I$, where $I(f \partial_x) := \int_\R f(x)dx$. So $\widehat{r(\psi)}$ is the differential operator given by
			\begin{equation}\label{eq: restr_cohom_to_vir}
				\widehat{r(\psi)}(f\partial_x) = f''' I + f (\partial_x \cdot \eta) + 2f^\prime \eta.
			\end{equation}
			Since $\widehat{r(\psi)}(f \partial_x) \in \X_c(\R)^\prime$ has compact support for any $f \in C^\infty(\R)$, we obtain by taking $f = 1$ in \eqref{eq: restr_cohom_to_vir} that $\partial_x \cdot \eta$ has compact support. Choosing subsequently $f(x) = x$ in \eqref{eq: restr_cohom_to_vir}, it follows that $\eta$ has compact support, and hence so does $\hat{\psi}_{\mathrm{vir}} = \widehat{r(\psi)} - \widehat{d_\g \eta}$. But the support of $\hat{\psi}_{\mathrm{vir}}$ is all of $\R$, which is not compact, a clear contradiction. \qedhere
		\end{enumerate}
	\end{proof}

	\newpage
	\appendix
	\section{Appendix}
	
	\subsection{Sheaves of distributions}
	Let $E \to M$ be a smooth vector bundle over the smooth manifold $M$. If $U \subseteq M$ is an open subset, we denote by $\Gamma_c(U, E)$ the locally convex vector space of smooth compactly supported sections of $\restr{E}{U} \to  U$, equipped with the natural LF-topology. Let $\Gamma_c(U, E)^\prime$ denote its continuous dual space. It is clear that the assignment $U \mapsto \Gamma_c(U, E)^\prime$ defines a presheaf $\Gamma_c^\prime$ w.r.t.\ the natural restriction maps. In the following, we show that $\Gamma_c^\prime$ actually defines an acyclic sheaf. Since we are considering the continuous dual space, we have to slightly extend the
	usual sheaf-theoretic arguments (such as \cite[{\S}V.1 Prop.\ 1.6 and 1.10]{Bredon_sheaf_theory}). 
	
	\begin{proposition}\label{prop: distr_acyclic_sheaf}
		$\Gamma_c^\prime$ is an acyclic sheaf.
	\end{proposition}
	\begin{proof}
		Let $\{U_\alpha\}_{\alpha \in \mc{I}}$ be a collection of open subsets of $M$ and define $U := \bigcup_{\alpha \in \mc{I}}U_\alpha$. Let $\{\chi_\alpha\}_{\alpha \in \mc{I}}$ be a partition of unity subordinate to the open cover $\{U_\alpha\}_{\alpha \in \mc{I}}$ of $U$ \cite[Thm.\ 2.23]{Lee_smooth_mfds}. Notice for $s \in \Gamma_c(U, E)$ that $\chi_\alpha s$ is non-zero for only finitely many $\alpha \in \mc{I}$, because $\SET{\supp \chi_\alpha}_{\alpha \in \mc{I}}$ is locally finite. To see that $\Gamma_c^\prime$ satisfies the locality axiom, suppose that $\lambda \in \Gamma_c(U, E)^\prime$ satisfies $\lambda_\alpha := \restr{\lambda}{\Gamma_c(U_\alpha, E)} = 0$ for all $\alpha \in \mc{I}$. Then $\lambda(s) = \sum_{\alpha \in \mc{I}}\lambda_\alpha(\chi_\alpha s) = 0$ for any $s \in \Gamma_c(U, E)$, so $\lambda = 0$. For the gluing axiom, take $\lambda_\alpha \in \Gamma_c(U_\alpha, E)^\prime$ for all $\alpha \in \mc{I}$ and suppose for any $\alpha, \beta \in \mc{I}$ that the restrictions of $\lambda_\alpha$ and $\lambda_\beta$ to $\Gamma_c(U_\alpha \cap U_\beta, E)$ coincide whenever $U_\alpha \cap U_\beta \neq 0$. Define $\lambda \in \Gamma_c(U, E)^\prime$ by $\lambda(s) := \sum_{\alpha \in \mc{I}}\lambda_\alpha(\chi_\alpha s)$ for $s \in \Gamma_c(U, E)$. Notice that $\lambda$ does indeed define a continuous functional on the LF-space $\Gamma_c(U, E)$ because $\SET{\supp \chi_\alpha}_{\alpha \in \mc{I}}$ is locally finite. If $s \in \Gamma_c(U_\beta, E)$ for some $\beta \in \mc{I}$, then $\chi_\alpha s \in \Gamma_c(U_\alpha \cap U_\beta, E)$ and consequently $\lambda_\alpha(\chi_\alpha s) = \lambda_\beta(\chi_\alpha s)$ for any $\alpha \in \mc{I}$. Hence $\lambda(s) = \sum_\alpha \lambda_\alpha(\chi_\alpha s) = \sum_\alpha\lambda_\beta(\chi_\alpha s) = \lambda_\beta(s)$. So $\restr{\lambda}{\Gamma_c(U_\beta, E)} = \lambda_\beta$ for any $\beta \in \mc{I}$. It follows that $\Gamma_c^\prime$ is a sheaf. We show next that it is fine (cf.\ \cite[Def.\ II.3.3]{Wells_book_cplx}). Assume henceforth that $U = M$. Define for any open set $V \subseteq M$ and $\alpha \in \mc{I}$ the linear map $\eta_{\alpha} : \Gamma_c(V, E)^\prime \to \Gamma_c(V, E)^\prime$ by $\eta_\alpha(\lambda)(s) := \lambda(\chi_\alpha s)$. This defines a morphism $\eta_\alpha : \Gamma_c^\prime \to \Gamma_c^\prime$ of sheaves. Since $\sum_{\alpha} \eta_\alpha(\lambda)(s) = \sum_\alpha \lambda(\chi_\alpha s) = \lambda(s)$ for any $s \in \Gamma_c(V, E)^\prime$, the sum being finite, we have $\sum_\alpha \eta_\alpha = 1$. Additionally, $\eta_\alpha$ vanishes on the stalk of the sheaf $\Gamma_c^\prime$ at $x$ for any $x$ in the open neighborhood $M \setminus \supp \chi_\alpha$ of $M\setminus U_\alpha$. So $\Gamma_c^\prime$ is fine and therefore acyclic \cite[II.\ Prop.\ 3.5 and Thm.\ 3.11]{Wells_book_cplx}.
	\end{proof}
	
	\subsection{Unitary equivalence of projective representations}
	\noindent
	Let $G$ be a locally convex Lie group with Lie algebra $\g$.
	
	\begin{definition}
		Suppose for $k \in \{1,2\}$ that $\mcD_k$ is a complex pre-Hilbert space, and let $\olpi_k : \g \to \pu(\mcD_k)$ be a projective unitary representation of $\g$ on $\mcD_k$. We say that $\olpi_1$ and $\olpi_2$ are unitarily equivalent if there exists a unitary operator $U : \mcD_1 \to \mcD_2$ such that $\olpi_2(\xi) = \overline{U} \,\olpi_1(\xi)\,\overline{U}^{-1}$ for all $\xi \in \g$, where $\overline{U} : \mrm{P}(\mcD_1) \to \mrm{P}(\mcD_2)$ is the descent of $U$ to the projective spaces. In this case, we write $\olpi_1 \cong \olpi_2$.
	\end{definition}
	
	\noindent
	The following is the projective analogue of \cite[Prop.\ 3.4]{BasNeeb_ProjReps}:
	
	\begin{proposition}\label{prop: projective_equivalence}
		Assume that $G$ is connected. For $k \in \{1,2\}$, let $(\overline{\rho}_k, \mc{H}_{k})$ be a smooth projective unitary representation of $G$ with derived representation $\overline{d\rho_k} : \g \to \pu(\mc{H}_{k}^\infty)$ on $\mc{H}_{k}^\infty$. Then
		$$ \overline{\rho}_1 \cong \overline{\rho}_2 \quad \iff \quad \overline{d\rho}_1 \cong \overline{d\rho}_2. $$
	\end{proposition}
	\begin{proof}
		\noindent
		Passing to the universal cover of $G$, which is a Lie group by \cite[Cor.\ II.2.4]{neeb_towards_lie}, we may and do assume that $G$ is $1$-connected. Let $U : \mc{H}_{1} \to \mc{H}_{2}$ be a unitary map, and let $\overline{U} : \mrm{P}(\mc{H}_{1}) \to \mrm{P}(\mc{H}_{2})$ be its descent to the projective spaces. Assume first that $\overline{U} \overline{\rho}_1(g) \overline{U}^{-1} = \overline{\rho}_2(g)$ for all $g \in G$. Then by \cite[Cor.\ 4.5]{BasNeeb_ProjReps}, we know that $\overline{\rho}_1$ and $\overline{\rho}_2$ correspond to the same central $\T$-extension $\circled{G}$ of $G$, up to isomorphism of central extensions. Let $\circled{\g}$ denote the Lie algebra of $\circled{G}$. Let the smooth unitary $\circled{G}$-representations $\rho_1$ and $\rho_2$ be lifts of $\overline{\rho}_1$ and $\overline{\rho}_2$ respectively. Then there exists a smooth character $\zeta: \circled{G} \to \U(1)$ such that $\rho_2(\circled{g}) = \zeta(\circled{g}) U \rho_1(\circled{g}) U^{-1}$ for all $\circled{g} \in \circled{G}$. This implies in particular that $U$ maps $\mc{H}_{1}^\infty$ onto $\mc{H}_{2}^\infty$. Differentiating the preceding equation at the identity of $\circled{G}$, it also follows that $d\rho_2(\circled{\xi}) = U d\rho_1(\circled{\xi}) U^{-1} + d\zeta(\circled{\xi})I$ for all $\circled{\xi} \in \circled{\g}$, where $I$ denotes the identity on $\mc{H}_{2}^\infty$. Hence $\overline{d\rho}_2(\xi) = \overline{U} \; \overline{d\rho}_1(\xi) \overline{U}^{-1}$ for all $\xi \in \g$. So $\overline{d\rho}_1 \cong \overline{d\rho}_2$. \\
		
		\noindent
		Assume conversely that $U$ maps $\mc{H}_{1}^\infty$ onto $\mc{H}_{2}^\infty$ and that $\overline{U} \; \overline{d\rho}_1(\xi)\, \overline{U}^{-1} = \overline{d\rho}_2(\xi)$ for all $\xi \in \g$. This implies that $\overline{d\rho}_1$ and $\overline{d\rho}_2$ induce isomorphic central $\R$-extension of $\g$, up to isomorphism. Since $G$ is $1$-connected, it follows using \cite[Cor.\ 7.15(i)]{Neeb_CE_inf_dim_Lie} that $\overline{\rho}_1$ and $\overline{\rho}_2$ induce the same central $\T$-extension $\circled{G}$ of $G$, up to isomorphism. Let the smooth unitary $\circled{G}$-representations $\rho_1$ and $\rho_2$ once again be lifts of $\overline{\rho}_1$ and $\overline{\rho}_2$, respectively. There exists a continuous linear map $\lambda : \circled{\g} \to \R$ such that 
		\begin{equation}\label{eq: nearly_equivariant}
			d\rho_2(\circled{\xi}) = U d\rho_1(\circled{\xi}) U^{-1} + i\lambda(\circled{\xi})I, \qquad \forall \circled{\xi} \in \circled{\g},
		\end{equation}	
		where $I$ denotes the identity on $\mc{H}_{2}^\infty$. Let $\psi \in \mc{H}_{1}^\infty$ and $\chi \perp \psi$. Let $(\chi_k)_{k \in \N}$ be a sequence in $\mc{H}_1^\infty$ such that $\chi_k \to \chi$ in $\mc{H}_1$. Let $\gamma : \R \to \circled{G}$ be a smooth path in $\circled{G}$ with $\gamma_0 = 1$ being the identity of $\circled{G}$, and let $\overline{\gamma} : \R \to G$ be its projection to $G$. Define the functions
		\begin{align*}
			f(t) &:= \langle U\chi, \rho_2(\gamma_t)U\rho_1(\gamma_t)^{-1} \psi\rangle,\\
			f_k(t) &:= \langle U\chi_k, \rho_2(\gamma_t)U\rho_1(\gamma_t)^{-1} \psi\rangle, \qquad t \in \R.
		\end{align*}
		Then $f_k \to f$ pointwise and $f_k$ is smooth for every $k \in \N$, as $U \chi_k \in \mc{H}_2^\infty$. Let $\gamma^\prime : \R\to \circled{\g}$ be the left-logarithmic derivative of $\gamma$, defined by $\gamma^\prime_s := \restr{\frac{d}{dt}}{t=s} \gamma_s^{-1}\gamma_t$ for $s \in \R$. Let $k \in \N$. Observe using \fref{eq: nearly_equivariant} that the derivative $f_k^\prime$ of $f_k$ satisfies
		\begin{align*}
			f_k^\prime(s)
			&= \langle U\chi_k, \rho_2(\gamma_s) d\rho_2(\gamma_s^\prime) U \rho_1(\gamma_s)^{-1} \psi \rangle - \langle U\chi_k,\rho_2(\gamma_s) U d\rho_1(\gamma_s^\prime)\rho_1(\gamma_s)^{-1} \psi\rangle \\
			&= i \lambda(\gamma_s^\prime)\langle U\chi_k, \rho_2(\gamma_s) U \rho_1(\gamma_s)^{-1} \psi\rangle\\
			&= i \lambda(\gamma_s^\prime)f_k(s), \qquad \forall s \in \R.
		\end{align*}
		We see that $f_k$ satisfies the ordinary differential equation $f_k'(s) = i \lambda(\gamma_s^\prime)f_k(s)$ with initial condition $f_k(0) = \langle \chi_k, \psi\rangle$. It follows that $f_k(t) = \langle \chi_k, \psi\rangle g(t)$, where $g(t) = e^{\int_0^t i \lambda(\gamma_s') \, \mrm{d}s}$. Consequently, $f(t) = \lim_k f_k(t) = \langle \chi, \psi\rangle g(t) = 0$ for all $t \in \R$.  Hence $\langle U\chi, \rho_2(\gamma_t) U \rho_1(\gamma_t)^{-1} \psi\rangle = 0$ for every $\chi \perp \psi$ and $t \in \R$. Thus $[\rho_2(\gamma_t) U \rho_1(\gamma_t)^{-1} \psi] = [U\psi]$ for all $t \in \R$ and $\psi \in \mc{H}_{1}^\infty$, which implies that $\overline{U} \, \overline{\rho}_1(\overline{\gamma}_t^{-1}) = \overline{\rho}_2(\overline{\gamma}_t^{-1}) \, \overline{U}$ for all $t \in \R$. Since $\circled{G}$ is a path-connected principal $\T$-bundle over $G$, it follows that $\overline{U} \, \overline{\rho}_1(g) = \overline{\rho}_2(g) \, \overline{U}$ for all $g \in G$. Thus $\overline{\rho}_1 \cong \overline{\rho}_2$.	
	\end{proof}
	
	\begin{corollary}\label{cor: kernels}
		Let $\overline{\rho} : G \to \PU(\mc{H})$ be a smooth projective unitary representation with derived representation $\overline{d\rho} : \g \to \pu(\mc{H}^\infty)$. Let $H$ be a connected Lie group with Lie algebra $\h$, and let $f : H \to G$ be smooth homomorphism of Lie groups. If $T_e(f)(\h) \subseteq \ker \overline{d\rho}$, then $f(H) \subseteq \ker \overline{\rho}$.
	\end{corollary}
	\begin{proof}
		By considering the pull-back of $\overline{\rho}$ along $f$, it suffices to consider the case where $H = G$ and $f = \id_G$. Thus, assume that $G$ is connected and that $\g \subseteq \ker \overline{d\rho}$. Then \Fref{prop: projective_equivalence} implies that $G \subseteq \ker \overline{\rho}$.
	\end{proof}
	
%
%
	
	\end{document}